\documentclass[a4paper, 11pt]{article}
\usepackage[utf8]{inputenc}
\usepackage[ngermanb, english]{babel}
\usepackage{amsmath}
\usepackage{amssymb}
\usepackage{amsthm}
\usepackage{euscript}
\usepackage{tikz-cd}
\usepackage{bm}
\usepackage{commath}
\usepackage{mathrsfs}
\usepackage{mathtools}
\usepackage{slashed}
\usepackage{tensor}
\usepackage{parskip}
\usepackage{epstopdf}
\usepackage{graphicx}
\usepackage{hhline}
\usepackage[nottoc]{tocbibind}
\usepackage{babelbib}
\usepackage{hyperref}
\usepackage[top=2.5cm, left=2.5cm, right=2.5cm, bottom=2.5cm]{geometry}

\theoremstyle{plain}
\newtheorem{thm}{Theorem}[section]

\newtheorem{prop}[thm]{Proposition}

\newtheorem{col}[thm]{Corollary}

\theoremstyle{definition}
\newtheorem{defn}[thm]{Definition}

\newtheorem{exmp}[thm]{Example}

\theoremstyle{remark}
\newtheorem{rem}[thm]{Remark}

\newcommand{\Q}{{\boldsymbol{\mathrm{Q}}}}
\newcommand{\BQ}{\mathcal{B}_\Q}

\newcommand{\FQ}{\mathcal{F}_\Q}
\newcommand{\dt}[1]{\operatorname{Det} \left ( #1 \right )}

\newcommand{\textfrac}[2]{#1 / #2}
\newcommand{\imaginary}{\mathrm{i}}

\newcommand{\enter}{\vspace{\baselineskip}}

\newcommand{\sctnbig}[1]{\Gamma \bigl ( #1 \bigr )}

\newcommand{\deDonder}{{d \negmedspace D \mspace{-2mu}}}

\newcommand{\Lorenz}{L}

\newcommand{\Lie}{\mathsterling}
\newcommand{\lie}{\ell}

\newcommand{\commutatorbig}[2]{\bigl [ #1 , #2 \bigr ]}
\newcommand{\setbig}[1]{\bigl \{ #1 \bigr \}}

\newsavebox{\foobox}

\makeatletter
\newcommand{\subalign}[1]{%
  \vcenter{%
    \Let@ \restore@math@cr \default@tag
    \baselineskip\fontdimen10 \scriptfont\tw@
    \advance\baselineskip\fontdimen12 \scriptfont\tw@
    \lineskip\thr@@\fontdimen8 \scriptfont\thr@@
    \lineskiplimit\lineskip
    \ialign{\hfil$\m@th\scriptstyle##$&$\m@th\scriptstyle{}##$\crcr
      #1\crcr
    }%
  }
}
\makeatother

\providecommand{\sectionref}[1]{Section~\ref{#1}}
\providecommand{\sectionsaref}[2]{Sections~\ref{#1} and \ref{#2}}

\providecommand{\eqnref}[1]{Equation~\eqref{#1}}

\providecommand{\eqnsaref}[2]{Equations~\eqref{#1} and \eqref{#2}}

\providecommand{\propsaref}[2]{Propositions~\ref{#1} and \ref{#2}}
\providecommand{\propscref}[2]{Propositions~\ref{#1}, \ref{#2}}

\providecommand{\ftnref}[1]{Footnote~\ref{#1}}

\providecommand{\defnsaref}[2]{Definitions~\ref{#1} and \ref{#2}}
\providecommand{\defnssaref}[3]{Definitions~\ref{#1}, \ref{#2} and \ref{#3}}

\providecommand{\thmsaref}[2]{Theoremata~\ref{#1} and \ref{#2}}

\providecommand{\exsaref}[2]{Examples~\ref{#1} and \ref{#2}}

\DeclareSymbolFont{extraitalic}      {U}{zavm}{m}{it}
\DeclareMathSymbol{\Qoppa}{\mathord}{extraitalic}{161}
\DeclareMathSymbol{\qoppa}{\mathord}{extraitalic}{162}
\DeclareMathSymbol{\Stigma}{\mathord}{extraitalic}{167}
\DeclareMathSymbol{\Sampi}{\mathord}{extraitalic}{165}
\DeclareMathSymbol{\sampi}{\mathord}{extraitalic}{166}
\DeclareMathSymbol{\stigma}{\mathord}{extraitalic}{168}

\title{\textsc{Symmetric Ghost Lagrange Densities for the Coupling of Gravity to Gauge Theories}}
\author{David Prinz\footnote{Department of Mathematics and Department of Physics at Humboldt University of Berlin and Department of Mathematics at University of Potsdam; prinz@\{math.hu-berlin.de, physik.hu-berlin.de, math.uni-potsdam.de\}}}
\date{November 15, 2025}

\begin{document}

\maketitle

\begin{abstract}
	We derive and present symmetric ghost Lagrange densities for the coupling of General Relativity to Yang--Mills theories. The graviton-ghost is constructed with respect to the linearized de Donder gauge fixing condition and the gauge ghost with respect to the covariant Lorenz gauge fixing condition. Both ghost Lagrange densities together with their accompanying gauge fixing Lagrange densities are obtained from the action of the diffeomorphism and gauge \emph{super-BRST differential} --- which we define as the composition of the BRST differential with its anti-BRST differential --- on suitable gauge fixing bosons. In addition, we introduce a \emph{total gauge fixing boson} and show that the complete symmetric ghost and gauge fixing Lagrange density can be generated thereof using the \emph{total super-BRST differential}. In particular, we generalize two earlier approaches for flat-spacetime Yang--Mills theories to General Relativity and covariant Yang--Mills theories: The original approach by Curci and Ferrari (1976), using the Faddeev--Popov method on non-linear gauge fixings, and the modern approach by Baulieu and Thierry-Mieg (1982), using BRST and anti-BRST symmetries with gauge fixing bosons.
\end{abstract}

\section{Introduction} \label{sec:introduction}

A central problem in the quantization of gauge theories is that the gauge symmetry needs to be broken in order to calculate the gauge boson propagator. To this end, a gauge fixing Lagrange density is added to the classical gauge theory Lagrange density: This now allows to obtain the gauge boson propagator as the inverse of the differential operator of the quadratic monomial. However, while this now produces well-defined tree-level expressions, a new problem arises at loop-level: Physical gauge bosons should only possess the experimentally verified transversal degrees of freedom. Unfortunately, gauge boson loops also produce a non-vanishing amplitude between transversal (i.e.\ physical) and longitudinal (i.e.\ unphysical) gauge boson modes. To overcome this issue, Feynman suggested to box and dismiss those diagrams to restore unitarity \cite{Feynman_QG}. This suggestion was then properly formulated by Faddeev and Popov by inventing so-called \emph{ghost} and \emph{antighost fields} with Grassmannian parity \cite{Faddeev_Popov}: These lead to fermionic particles with integer spin and thus violating the spin-statistic theorem --- however, when considered only as virtual particles, transversality of the perturbative expansion is successfully restored. This Faddeev--Popov construction has then been further embedded into an even more general setup using homological algebra into what is now called \emph{BRST symmetry} \cite{Becchi_Rouet_Stora_1,Becchi_Rouet_Stora_2,Becchi_Rouet_Stora_3,Tyutin} and the \emph{BV formalism} \cite{Batalin_Vilkovisky_1,Batalin_Vilkovisky_2}.

The outline of the present analysis is as follows: Given a gauge fixing condition, the ghost Lagrange density in the Faddeev--Popov construction (which is typically also used in the more general frameworks) is constructed such that the ghost field satisfies the residual gauge transformations as equations of motion, with the antighost field acting as the corresponding Lagrange multiplier field. While this construction has a clear interpretation in terms of its dynamics, its shortcoming is that the antighost is not the antiparticle of the ghost: This point has been first addressed by Curci and Ferrari using non-linear gauge fixing functionals \cite{Curci_Ferrari}. Then, Baulieu and Thierry-Mieg constructed a symmetric ghost Lagrange density, as well as a homotopy between different ghost Lagrange density constructions, for pure Quantum Yang--Mills theory on the basis of BRST and anti-BRST symmetries \cite{Baulieu_Thierry-Mieg}. It is precisely the aim of this article to generalize their construction to (effective) Quantum General Relativity and covariant Quantum Yang--Mills theory. Notably, the present analysis is building on the author's previous article \cite{Prinz_5} and the references therein, where the general setup has been constructed and studied: Concretely, in this article we provided a detailed mathematical introduction to the appearing constructions and particle fields using graded supergeometry. Specifically, the BRST operator can be understood as a cohomological vector field on the graded superbundle of particle fields and the gauge fixing fermion as a local functional of particle fields in ghost-degree minus one. Then, the gauge fixing and ghost Lagrange densities can be generated by acting with the BRST operator on such a gauge fixing fermion, which produces a local functional in ghost-degree zero. Notably, since the so-constructed gauge fixing and ghost Lagrange densities are BRST-exact, they do not contribute to the corresponding zeroth cohomology, which describes the physical degrees of freedom. In the present article, we start with a so-called gauge fixing boson, which is a local functional of particle fields in ghost-degree zero. Then, the application of the anti-BRST operator turns this into a gauge fixing fermion by transforming it into ghost-degree minus one and we can continue as before. However, this more involved construction has now a number of interesting consequences: First, the gauge fixing condition appears naturally and is the \emph{optimal gauge fixing condition}, as is studied in \cite{Prinz_7}. Furthermore, the obtained ghost Lagrange density is symmetric with respect to the ghost conjugation, such that the antighost is the antiparticle of the ghost. This leads to interesting symmetries and cancellations of longitudinal modes, which is also studied in \cite{Prinz_7}. Finally, we can also construct a homotopy between different ghost constructions, which can be seen as a \emph{ghost parameter} in addition to the well-known \emph{gauge fixing parameter}.

More precisely, given a gauge field \(\boldsymbol{\varphi}\) with coupling constant \(\alpha\), a corresponding infinitesimal gauge transformation \(\delta_Z \boldsymbol{\varphi}\) with respect to a Lie algebra valued vector field \(Z\) and a chosen gauge fixing functional \(\operatorname{GF} \left ( \boldsymbol{\varphi} \right )\). Let furthermore \(\theta\) and \(\overline{\theta}\) denote the corresponding ghost and antighost fields, \(\beta\) the Lautrup--Nakanishi auxiliary field \cite{Nakanishi,Lautrup} and \(\lambda\) the gauge fixing parameter. Then, the gauge fixing and Faddeev--Popov ghost Lagrange density reads
\begin{equation} \label{eqn:gf-fp-ghost}
		\mathcal{L}_\text{GF-FP-Ghost} \coloneq \left ( - \frac{1}{2 \alpha^2 \lambda} \operatorname{GF} \left ( \boldsymbol{\varphi} \right )^2 + \overline{\theta} \cdot \operatorname{GF} \left ( \delta_\theta \boldsymbol{\varphi} \right ) \right ) \dif V_g \, ,
\end{equation}
where ` \(\cdot\) ' denotes a scalar product on the corresponding Lie algebra and \(\dif V_g\) denotes the Riemannian volume form, see below for the definition. In particular, this can be generated from the following \emph{gauge fixing fermion}
\begin{equation}
	\chi \coloneq \overline{\theta} \cdot \left ( \frac{1}{2 \alpha} \operatorname{GF} \left ( \boldsymbol{\varphi} \right ) + \frac{1}{4} \beta \right ) \dif V_g \, ,
\end{equation}
which is a functional in ghost degree minus one, via the action of the corresponding BRST operator \(S \coloneq \delta_\theta\) (see \defnsaref{defn:diffeomorphism_brst_operator}{defn:gauge_brst_operator}), i.e.\ \(\mathcal{L}_\text{GF-FP-Ghost} \coloneq S \chi\), cf.\ \cite[Propositions 3.6 and 4.5]{Prinz_5}. This construction has the advantage of being rather simple to calculate. In addition, it provides an immediate interpretation for the equations of motion of the ghost and antighost fields: While the ghost field is constructed to satisfy residual gauge transformations as equations of motion, the antighost field is acting as the corresponding Lagrange multiplier. Unfortunately, the apparent asymmetry between the ghost and antighost results in an intransparent relationship between these two fields. This becomes in particular relevant when analyzing the longitudinal and transversal contributions of the corresponding Feynman integrals. Specifically, they can be understood via the so-called \emph{cancellation identities}, which establish pairwise cancellations of longitudinally contracted Feynman graphs in the perturbative expansion, cf.\ \cite{Prinz_7,tHooft_Veltman,Citanovic,Kissler_Kreimer,Gracey_Kissler_Kreimer,Kissler,Prinz_9}. Notably, when combined with the parametric representation of Feynman integrals, they can be also summarized into the definition of a third graph polynomial --- the so-called \emph{Corolla polynomial} --- cf.\ \cite{Kreimer_Yeats,Kreimer_Sars_vSuijlekom,Sars_PhD,Prinz_1,Golz_PhD}. Thus, this article is devoted to a proper derivation of symmetric ghost Lagrange densities using appropriate gauge fixing bosons, BRST and anti-BRST operators. More precisely, the resulting Lagrange densities of \emph{Quantum Gauge Theories (QGT)}\footnote{On the level of the Lagrange density, we use the word \emph{quantum} to indicate that the gauge fixing and ghost Lagrange densities have been added, such that a perturbative Feynman graph expansion is possible, i.e.\ the propagator of the gauge boson can be calculated and loop-level expressions are transversal.} with a symmetric ghost construction are Hermitian with respect to the ghost conjugation \(\dagger\) of \defnref{defn:ghost-conjugation}, i.e.\
\begin{subequations}
\begin{align}
	\mathcal{L}_\text{QGT}^\dagger & \equiv \mathcal{L}_\text{QGT} \, , \label{eqn:hermitian-qgt-lagrange-density}
	\intertext{where the ghost conjugation is defined as the following Hermitian ghost-grading reversal, interchanging ghosts and antighosts, i.e.}
	\theta^\dagger & \coloneq \overline{\theta} \, , \\
	\overline{\theta}^\dagger & \coloneq \theta
	\intertext{and}
	\beta^\dagger & \coloneq - \beta - \alpha \lambda \commutatorbig{ \, \overline{\theta}}{\theta \, } \, .
	\intertext{We remark that it will be later also convenient to define the shifted anti-Hermitian Lautrup--Nakanishi auxiliary field, given via}
	\beta^\prime & \coloneq \beta - \frac{\alpha \lambda}{2} \commutatorbig{ \, \overline{\theta}}{\theta \, } \, ,
	\intertext{such that}
	{\beta^\prime}^\dagger & \phantom{:} \equiv - \beta^\prime \, .
\end{align}
\end{subequations}
Thus, in the situation of \eqnref{eqn:hermitian-qgt-lagrange-density} the antighost is actually the antiparticle of ghost. Such Lagrange densities were first constructed by Curci and Ferrari using non-linear gauge fixing conditions for Yang--Mills theories on a flat spacetime \cite{Curci_Ferrari}. Then, their construction was formalized and generalized by Baulieu and Thierry-Mieg using BRST and anti-BRST symmetries \cite{Baulieu_Thierry-Mieg}. In this article, we generalize this construction to General Relativity and covariant Yang--Mills theories, that is Yang--Mills theories on curved spacetimes via an appropriate coupling to General Relativity. More precisely, we start with a so-called \emph{gauge fixing boson}\footnote{In the gravitational case it is even possible to be linear in the graviton field, cf.\ \eqnref{eqn:qym-gauge-fixing-boson-introduction} and \remref{rem:qgr-gauge-fixing-boson}.}
\begin{equation}
	W \coloneq - \frac{\lambda}{2} \left ( \boldsymbol{\varphi}^2 - \overline{\theta} \cdot \theta \right ) \dif V_g \, ,
\end{equation}
which is a functional in ghost degree zero. From this, we construct its associated gauge fixing fermion \(\omega\) via the action of the anti-BRST operator \(\overline{S}\), i.e.\ \(\omega \coloneq \overline{S} W\). Then, the corresponding gauge fixing and ghost Lagrange density \(\mathcal{L}_\text{GF-Ghost}\) is given via the action of the BRST operator \(S\) on \(\omega\), i.e.\ \(\mathcal{L}_\text{GF-Ghost} \coloneq S \omega\). Thus, finally we obtain:
\begin{equation}
	\mathcal{L}_\text{GF-Sym-Ghost} \coloneq \mathcal{S} W \, ,
\end{equation}
where we have introduced the \emph{super-BRST operator} \(\mathcal{S}\), defined via
\begin{equation}
	\mathcal{S} \coloneq S \circ \overline{S} \equiv - \overline{S} \circ S \equiv \frac{1}{2} \left ( S \circ \overline{S} - \overline{S} \circ S \right ) \, .
\end{equation}
The equivalent expressions are due to the property \(\commutatorbig{S}{\overline{S}} \equiv 0\), cf.\ \cite[Corollaries 3.4 and 4.4]{Prinz_5}.\footnote{We emphasize that we use the symbol \(\left [ \, \cdot \, , \cdot \, \right ]\) for the supercommutator: In particular, it denotes the anticommutator if both arguments are odd.} We remark that the so-constructed Lagrange density still contains the corresponding Lautrup--Nakanishi auxiliary field: This field is neither Hermitian nor anti-Hermitian with respect to the ghost conjugation. However, it can be shifted to become anti-Hermitian. Once it is eliminated via its equations of motion after this shift, we obtain the following symmetric setting:
\begin{equation} \label{eqn:gf-sym-ghost}
\begin{split}
	\mathcal{L}_\text{GF-Sym-Ghost} & \coloneq \left ( - \frac{1}{2 \alpha^2 \lambda} \operatorname{GF} \left ( \boldsymbol{\varphi} \right )^2 + \frac{1}{2} \left ( \overline{\theta} \cdot \operatorname{GF} \left ( \delta_\theta \boldsymbol{\varphi} \right ) + \operatorname{GF} \left ( \delta_{\overline{\theta}} \boldsymbol{\varphi} \right ) \cdot \theta \right ) \right ) \dif V_g \\ & \phantom{\coloneq} + \frac{\alpha^2 \lambda}{16} \left ( \commutatorbig{ \, \overline{\theta}}{\overline{\theta} \, } \cdot \commutatorbig{ \, \theta}{\theta \, } \right ) \dif V_g
\end{split}
\end{equation}
We remark that the gauge fixing functional \(\operatorname{GF} \left ( \boldsymbol{\varphi} \right )\) for a given gauge theory and gauge fixing boson is now determined to be an \emph{optimal gauge fixing}, a notion that has been introduced by the author in the follow-up article \cite{Prinz_7}: In particular, there it is shown that for General Relativity this is given as the (linearized) de Donder gauge fixing condition and for Yang--Mills theory this is given as the (covariant) Lorenz gauge fixing condition. Additionally, we highlight the newly appearing four-ghost-interaction in addition to the symmetrized Faddeev--Popov construction. Notably, in \cite{Prinz_5,Prinz_7} and the present article, we use the convention that the longitudinal mode of the gauge boson propagator as well as the ghost propagator are both scaled by the gauge fixing parameter \(\lambda\), which in this convention also appears as a prefactor of said four-ghost-interaction, cf.\ \cite{Baulieu_Thierry-Mieg} for comparison. More generally, this construction can then be embedded into a homotopy between different ghost constructions by introducing a \emph{ghost parameter} \(\varrho\), resulting in the following Lagrange density:
\begin{equation} \label{eqn:gf-hom-ghost}
\begin{split}
	\mathcal{L}_\text{GF-Hom-Ghost} \left ( \varrho \right ) & \coloneq \left ( - \frac{1}{2 \alpha^2 \lambda} \operatorname{GF} \left ( \boldsymbol{\varphi} \right )^2 + \frac{1}{2} \left ( \left ( 1 - \varrho \right ) \overline{\theta} \cdot \operatorname{GF} \left ( \delta_\theta \boldsymbol{\varphi} \right ) + \varrho \operatorname{GF} \left ( \delta_{\overline{\theta}} \boldsymbol{\varphi} \right ) \cdot \theta \right ) \right ) \dif V_g \\ & \phantom{\coloneq} + \frac{\alpha^2 \lambda \, \varrho \left ( 1 - \varrho \right )}{4} \left ( \commutatorbig{ \, \overline{\theta}}{\overline{\theta} \, } \cdot \commutatorbig{ \, \theta}{\theta \, } \right ) \dif V_g
\end{split}
\end{equation}
Observe that \(\varrho = 0\) corresponds to the Faddeev--Popov construction, displayed in \eqnref{eqn:gf-fp-ghost}, \(\varrho = \textfrac{1}{2}\) to the symmetric setting, displayed in \eqnref{eqn:gf-sym-ghost}, and \(\varrho = 1\) to the opposed Faddeev--Popov construction, i.e.\ the ghost conjugation of \eqnref{eqn:gf-fp-ghost}.

More specifically, we consider General Relativity, given via the Lagrange density:
\begin{equation} \label{eqn:gr_lagrange_density_introduction}
	\mathcal{L}_\text{GR} \coloneq - \frac{1}{2 \varkappa^2} R \dif V_g
\end{equation}
In particular, we consider the metric expansion \(g_{\mu \nu} \equiv \eta_{\mu \nu} + \varkappa h_{\mu \nu}\), where \(h_{\mu \nu}\) is the graviton field and \(\varkappa \coloneq \sqrt{\kappa}\) the graviton coupling constant (with \(\kappa \coloneq 8 \pi G\) being the Einstein gravitational constant). In addition, \(R \coloneq g^{\nu \sigma} \tensor{R}{^\mu _\nu _\mu _\sigma}\) is the Ricci scalar, where \(\tensor{R}{^\rho _\sigma _\mu _\nu} \coloneq \partial_\mu \tensor{\Gamma}{^\rho _\nu _\sigma} - \partial_\nu \tensor{\Gamma}{^\rho _\mu _\sigma} + \tensor{\Gamma}{^\rho _\mu _\lambda} \tensor{\Gamma}{^\lambda _\nu _\sigma} - \tensor{\Gamma}{^\rho _\nu _\lambda} \tensor{\Gamma}{^\lambda _\mu _\sigma}\) is the Riemann tensor with \(\tensor{\Gamma}{^\rho _\mu _\nu} = g^{\rho \sigma} \left ( \partial_\mu g_{\sigma \nu} + \partial_\nu g_{\mu \sigma} - \partial_\sigma g_{\mu \nu} \right ) / 2\) the Christoffel symbol. Furthermore, \(\dif V_g \coloneq \sqrt{- \dt{g}} \dif t \wedge \dif x \wedge \dif y \wedge \dif z\) denotes the Riemannian volume form and \(\dif V_\eta \coloneq \dif t \wedge \dif x \wedge \dif y \wedge \dif z\) the Minkowskian volume form. Then we obtain the following result in \propref{prop:symmetric-ghost-qgr}: Starting with the gauge fixing boson
\begin{equation} \label{eqn:qgr-gauge-fixing-boson-introduction}
	F \coloneq - \frac{\zeta}{4} \left ( \frac{1}{\varkappa} \eta^{\mu \nu} h_{\mu \nu} - \overline{C}^\rho C_\rho \right ) \dif V_\eta
\end{equation}
we obtain the following symmetric gauge fixing and ghost Lagrange density, where \(\zeta\) denotes the de Donder gauge fixing parameter and \(\mathcal{P} \coloneq P \circ \overline{P}\) is the diffeomorphism super-BRST operator:\footnote{We remark that the prefactor of the four-ghost-interaction is \(\textfrac{1}{32}\) times the product of the commutators, cf.\ \eqnref{eqn:gf-sym-ghost} --- however, since the resulting four terms can be summarized into a single term, the prefactor then becomes the indicated \(\textfrac{1}{8}\).}
\begin{equation} \label{eqn:Lagrange-GR-GF-Ghost}
\begin{split}
	\mathcal{L}_\text{GR-GF-Sym-Ghost} & \coloneq \mathcal{P} F \\
	& \phantom{:} \equiv \frac{1}{2 \zeta} \left ( - \frac{1}{2 \varkappa^2} \eta^{\mu \nu} \deDonder^{(1)}_\mu \deDonder^{(1)}_\nu + \eta^{\mu \nu} \bigl ( \partial_\mu \overline{C}^\rho \bigr ) \bigl ( \partial_\nu C_\rho \bigr ) \right ) \dif V_\eta \\
	& \phantom{\coloneq} + \frac{1}{4} \eta^{\mu \nu} \left ( \bigl ( \partial_\rho \overline{C}^\rho \bigr ) \bigl ( \Gamma_{\sigma \mu \nu} C^\sigma \bigr ) - 2 \bigl ( \partial_\mu \overline{C}^\rho \bigr ) \bigl ( \Gamma_{\sigma \rho \nu} C^\sigma \bigr ) \right ) \dif V_\eta \\
	& \phantom{\coloneq} + \frac{1}{4} \eta^{\mu \nu} \left ( \bigl ( \Gamma_{\sigma \mu \nu} \overline{C}^\sigma \bigr ) \bigl ( \partial_\rho C^\rho \bigr ) - 2 \bigl ( \Gamma_{\sigma \rho \nu} \overline{C}^\sigma \bigr ) \bigl ( \partial_\mu C^\rho \bigr ) \right ) \dif V_\eta \\
	& \phantom{\coloneq} + \frac{\varkappa^2 \zeta}{8} \eta_{\mu \nu} \left ( \overline{C}^\rho \bigl ( \partial_\rho \overline{C}^\mu \bigr ) \right ) \left ( C^\sigma \bigl ( \partial_\sigma C^\nu \bigr ) \right ) \dif V_\eta
\end{split}
\end{equation}
Here, \(\deDonder^{(1)}_\mu \coloneq \eta^{\rho \sigma} \Gamma_{\mu \rho \sigma} \equiv 0\) is the linearized de Donder gauge fixing functional with \(\Gamma_{\mu \rho \sigma} \coloneq \textfrac{\varkappa \left ( \partial_\rho h_{\mu \sigma} + \partial_\sigma h_{\rho \mu} - \partial_\mu h_{\rho \sigma} \right )}{2}\) and \(C_\mu\) and \(\overline{C}^\mu\) are the graviton-ghost and graviton-antighost fields, respectively. Finally, the Lagrange density for (effective) Quantum General Relativity is then given as the sum of the two:
\begin{equation}
	\mathcal{L}_\text{QGR} \coloneq \mathcal{L}_\text{GR} + \mathcal{L}_\text{GR-GF-Ghost}
\end{equation}
In addition, in \thmref{thm:symmetric-ghost-homotopy-qgr} we also construct a homotopy in the sense of \eqnref{eqn:gf-hom-ghost} that continuously interpolates between the corresponding Faddeev--Popov construction, cf.\ \cite[Corollary 3.7]{Prinz_5}, the symmetric setting of \propref{prop:symmetric-ghost-qgr} and the opposed Faddeev--Popov construction, where we introduce the \emph{graviton-ghost parameter} \(\varsigma\):
\begin{equation}
\begin{split}
	\mathcal{L}_\textup{GR-GF-Ghost} \left ( \varsigma \right ) & \coloneq \frac{1}{2 \zeta} \left ( - \frac{1}{2 \varkappa^2} \eta^{\mu \nu} \deDonder^{(1)}_\mu \deDonder^{(1)}_\nu + \eta^{\mu \nu} \bigl ( \partial_\mu \overline{C}^\rho \bigr ) \bigl ( \partial_\nu C_\rho \bigr ) \right ) \dif V_\eta \\
	& \phantom{\coloneq} + \frac{\left ( 1 - \varsigma \right )}{2} \eta^{\mu \nu} \left ( \bigl ( \partial_\rho \overline{C}^\rho \bigr ) \bigl ( \Gamma_{\sigma \mu \nu} C^\sigma \bigr ) - 2 \bigl ( \partial_\mu \overline{C}^\rho \bigr ) \bigl ( \Gamma_{\sigma \rho \nu} C^\sigma \bigr ) \right ) \dif V_\eta \\
	& \phantom{\coloneq} + \frac{\varsigma}{2} \eta^{\mu \nu} \left ( \bigl ( \Gamma_{\sigma \mu \nu} \overline{C}^\sigma \bigr ) \bigl ( \partial_\rho C^\rho \bigr ) - 2 \bigl ( \Gamma_{\sigma \rho \nu} \overline{C}^\sigma \bigr ) \bigl ( \partial_\mu C^\rho \bigr ) \right ) \dif V_\eta \\
	& \phantom{\coloneq} + \frac{\varkappa^2 \zeta \varsigma \left ( 1 - \varsigma \right )}{2} \eta_{\mu \nu} \left ( \overline{C}^\rho \bigl ( \partial_\rho \overline{C}^\mu \bigr ) \right ) \left ( C^\sigma \bigl ( \partial_\sigma C^\nu \bigr ) \right ) \dif V_\eta
\end{split}
\end{equation}
Moreover, we discuss other valid choices for gauge fixing bosons in \remref{rem:qgr-gauge-fixing-boson}.

Additionally, we consider Yang--Mills theory, given via the Lagrange density:
\begin{equation} \label{eqn:ym_lagrange_density_introduction}
	\mathcal{L}_\text{YM} \coloneq - \frac{1}{4 \mathrm{g}^2} \delta_{a b} g^{\mu \nu} g^{\rho \sigma} F^a_{\mu \rho} F^b_{\nu \sigma} \dif V_g
\end{equation}
Here, \(F^a_{\mu \nu} \coloneq \mathrm{g} \bigl ( \partial_\mu A^a_\nu - \partial_\nu A^a_\mu \bigr ) - \mathrm{g}^2 \tensor{f}{^a _b _c} A^b_\mu A^c_\nu\) is the local curvature form of the gauge boson \(A^a_\mu\) and \(\mathrm{g}\) the gauge boson coupling constant. Furthermore, \(\dif V_g \coloneq \sqrt{- \dt{g}} \dif t \wedge \dif x \wedge \dif y \wedge \dif z\) denotes again the Riemannian volume form. Then we obtain the following result in \propref{prop:symmetric-ghost-qym}: Starting with the gauge fixing boson
\begin{equation} \label{eqn:qym-gauge-fixing-boson-introduction}
	G \coloneq - \frac{\xi}{2} \left ( \delta_{a b} g^{\mu \nu} A^a_\mu A^b_\nu - \overline{c}_a c^a \right ) \dif V_g
\end{equation}
we obtain the following symmetric gauge fixing and ghost Lagrange density, where \(\xi\) denotes the Lorenz gauge fixing parameter and \(\mathcal{Q} \coloneq Q \circ \overline{Q}\) is the gauge super-BRST operator:
\begin{equation} \label{eqn:Lagrange-YM-GF-Ghost}
	\begin{split}
		\mathcal{L}_\text{YM-GF-Sym-Ghost} & \coloneq \mathcal{Q} G \\
		& \phantom{:} \equiv \frac{1}{\xi} \left ( - \frac{1}{2 \mathrm{g}^2} \delta_{ab} \Lorenz^a \Lorenz^b + g^{\mu \nu} \left ( \partial_\mu \overline{c}_a \right ) \left ( \partial_\nu c^a \right ) \right ) \dif V_g \\
		& \phantom{\coloneq} + \frac{\mathrm{g}}{2} g^{\mu \nu} \tensor{f}{^a _b _c} \left ( \bigl ( \partial_\mu \overline{c}_a \bigr ) \bigl ( c^b A^c_\nu \bigr ) + \bigl ( \overline{c}_a A^b_\nu \bigr ) \bigl ( \partial_\mu c^c \bigr ) \right ) \dif V_g \\
		& \phantom{\coloneq} + \frac{\mathrm{g}^2 \xi}{16} f^{a b c} f_{ade} \overline{c}_b \overline{c}_c c^d c^e \dif V_g
	\end{split}
\end{equation}
Here, \(\Lorenz^a \coloneq \mathrm{g} g^{\mu \nu} \bigl ( \nabla^{TM}_\mu A^a_\nu \bigr ) \equiv 0\) is the covariant Lorenz gauge fixing functional and \(c^a\) and \(\overline{c}_a\) are the gauge ghost and gauge antighost fields, respectively. Finally, the Lagrange density for Quantum Yang--Mills theory is then given as the sum of the two:
\begin{equation}
	\mathcal{L}_\text{QYM} \coloneq \mathcal{L}_\text{YM} + \mathcal{L}_\text{YM-GF-Ghost}
\end{equation}
In addition, in \thmref{thm:symmetric-ghost-homotopy-qym} we also construct a homotopy in the sense of \eqnref{eqn:gf-hom-ghost} that continuously interpolates between the corresponding Faddeev--Popov construction, cf.\ \cite[Corollary 4.6]{Prinz_5}, the symmetric setting of \propref{prop:symmetric-ghost-qym} and the opposed Faddeev--Popov construction, where we introduce the \emph{gauge ghost parameter} \(\vartheta\):
\begin{equation}
\begin{split}
	\mathcal{L}_\textup{YM-GF-Ghost} \left ( \vartheta \right ) & \coloneq \frac{1}{\xi} \left ( - \frac{1}{2 \mathrm{g}^2} \delta_{ab} \Lorenz^a \Lorenz^b + g^{\mu \nu} \left ( \partial_\mu \overline{c}_a \right ) \left ( \partial_\nu c^a \right ) \right ) \dif V_g \phantom{\bigg \vert} \\
	& \phantom{\coloneq} + \frac{\mathrm{g}}{2} g^{\mu \nu} \tensor{f}{^a _b _c} \left ( \left ( 1 - \vartheta \right ) \bigl ( \partial_\mu \overline{c}_a \bigr ) \bigl ( c^b A^c_\nu \bigr ) + \vartheta \bigl ( \overline{c}_a A^b_\nu \bigr ) \bigl ( \partial_\mu c^c \bigr ) \right ) \dif V_g \phantom{\bigg \vert} \\
	& \phantom{\coloneq} + \frac{\mathrm{g}^2 \xi \vartheta \left ( 1 - \vartheta \right )}{4} f^{a b c} f_{ade} \overline{c}_b \overline{c}_c c^d c^e \dif V_g \phantom{\bigg \vert}
\end{split}
\end{equation}
Moreover, to complement the theoretical constructions and insights, we work out the specific cases of a \(\operatorname{SU}(2)\) gauge group and a Schwarzschild spacetime in \exsaref{exmp:qym-su-two}{exmp:qym-schwarzschild}.

Finally, for the coupling of (effective) Quantum General Relativity to Quantum Yang--Mills theory, we obtain the following results: Given the \emph{total BRST operator} \(D \coloneq P + Q\) from \cite[Theorem 5.1]{Prinz_5} and the \emph{total anti-BRST operator} \(\overline{D} \coloneq \overline{P} + \overline{Q}\) from \cite[Corollary 5.2]{Prinz_5} and let \(\mathcal{D} \coloneq D \circ \overline{D}\) be the \emph{total super-BRST operator} from \defnref{defn:super-brst-operator}. Then we can generate the complete gauge fixing and ghost Lagrange density using a \emph{total gauge fixing boson} \(Y \coloneq F + G\) via \(\mathcal{D} Y\), cf.\ \thmref{thm:total_gf_ghost_lagrange_density}. In addition, we also obtain a double-homotopy by adding the two individual homotopies, cf.\ \colref{col:total-setup-double-homotopy}.

We refer to \cite{Prinz_2,Prinz_4,Prinz_8} for more detailed introductions to (effective) Quantum General Relativity coupled to Quantum Yang--Mills theories using the same conventions. In addition, we refer to \cite{Prinz_5} for the introduction of the diffeomorphism-gauge BRST double complex and its corresponding anti-BRST complex. We will use the ghost Lagrange densities from this article in \cite{Prinz_7} to study the cancellation identities for (effective) Quantum General Relativity coupled to the Standard Model. This provides an important ingredient to study the renormalization of Quantum Gauge Theories, cf.\ \cite{Prinz_9,Kreimer_QG1,vSuijlekom_BV,Prinz_3}. Moreover, we refer the interested reader to the introductory texts on BRST cohomology and the BV formalism \cite{Barnich_Brandt_Henneaux_1,Mnev,Wernli}, the historical overview \cite{Becchi} and earlier works in a similar direction \cite{Baulieu_Thierry-Mieg,Nakanishi_Ojima,Faizal,Baulieu_Bellon_1,Baulieu_Bellon_2,Shestakova,Barnich_Brandt_Henneaux_2,Upadhyay}.

This article is related to the author's dissertation \cite{Prinz_PhD}.

\section{General setup} \label{sec:general-setup}

We start this article with a brief summary of the diffeomorphism-gauge BRST double complex, which was introduced in \cite{Prinz_5}: This includes the definitions of the diffeomorphism, gauge and total BRST and anti-BRST operators. Then we introduce the ghost conjugations and discuss its action on the diffeomorphism and gauge Lautrup--Nakanishi auxiliary fields. In particular, we shift them such that they become anti-Hermitian with respect to their associated ghost conjugation. We refer to \cite[Section 2]{Prinz_5} for a detailed mathematical introduction to the fields as well as the above mentioned BRST and anti-BRST operators using graded supergeometry with cohomological and homological vector fields, respectively. Finally, we remark that the following constructions work on general spacetime manifolds with any dimension and structure constants of arbitrary compact semisimple Lie algebras.

\enter

\begin{defn}[Spacetime] \label{def:spacetime}
	Let \((M,g)\) be a \(d\)-dimensional Lorentzian manifold. We call \((M,g)\) a spacetime, if it is smooth, connected and time-orientable.
\end{defn}

\enter

\begin{defn}[Spacetime-matter bundle] \label{defn:spacetime-matter_bundle}
	Let \((M, g)\) be a \(d\)-dimensional spacetime and \(G\) a compact and semisimple Lie group with Lie algebra \(\mathfrak{g}\).\footnote{We remind the reader that this implies that its \emph{Killing form} is negative-definite: Thus, it ensures that the Yang--Mills Lagrange density is non-negative and that there are no additional zero-modes.} Then we define the spacetime-matter bundle of (effective) Quantum General Relativity coupled to Quantum Yang--Mills theory as the \(\mathbb{Z}^2\)-graded super bundle \(\beta_\Q \colon \BQ \to M\) (some further applications might require the bundle to be trivial, but for the constructions in this article and the local considerations in physics a general bundle is fine), where \(\BQ \coloneq M \times_M \mathcal{V}_\Q\) is the fiber product over \(M\) with
	\begin{equation}
	\begin{split}
		\mathcal{V}_\Q & \coloneq \left ( \operatorname{Sym}^2_\mathbb{R} \left ( T^* M \right ) \right )^{\times 3} \times \Bigl ( T^* [1,0] M \oplus T [-1,0] M \oplus T M \Bigr ) \\ & \phantom{\coloneq} \times \left ( T^*M \otimes_\mathbb{R} \mathfrak{g} \right ) \times \Bigl ( \mathfrak{g} [0,1] \oplus \mathfrak{g}^* [0,-1] \oplus \mathfrak{g}^* \Bigr ) \, , \label{eqn:spacetime-matter_bundle}
	\end{split}
	\end{equation}
	where we have the following bundles:
	\begin{itemize}
		\item Metric, background metric and graviton field as a section in the triple Cartesian product \(\bigl ( \operatorname{Sym}^2_\mathbb{R} \left ( T^* M \right ) \bigr )^{\times 3} \coloneq \bigtimes_{m = 1}^3 \bigl ( \operatorname{Sym}^2_\mathbb{R} \left ( T^* M \right ) \bigr )\), where \(\operatorname{Sym}^2_\mathbb{R} \left ( T^* M \right ) \coloneq \left ( T^* M \otimes_\mathbb{R} T^* M \right ) / \mathbb{Z}_2\) is the symmetrized tensor product
		\item Graviton-ghost as a section in \(T^* [1,0] M\)
		\item Graviton-antighost as a section in \(T [-1,0] M\)
		\item Graviton-Lautrup--Nakanishi field as a section in \(T M\)
		\item Gauge bosons as a section in \(T^*M \otimes_\mathbb{R} \mathfrak{g}\)
		\item Gauge ghost as a section in the bundle with fiber \(\mathfrak{g} [0,1]\)
		\item Gauge antighost as a section in the bundle with fiber \(\mathfrak{g}^* [0,-1]\)
		\item Gauge Lautrup--Nakanishi field as a section in the bundle with fiber \(\mathfrak{g}^*\)
	\end{itemize}
	Here, the ghosts are odd sections of either graviton-ghost degree \(\pm 1\) or gauge ghost degree \(\pm 1\), respectively, cf.\ \cite[Section 2]{Prinz_5} for the mathematical background and technical setup.
\end{defn}

\enter

\begin{defn}[Sheaf of particle fields] \label{defn:sheaf-of-particle-fields}
	Let \((M, g)\) be a spacetime with topology \(\mathcal{T}_M\) and \(\beta_\Q \colon \BQ \to M\) the spacetime-matter bundle from \defnref{defn:spacetime-matter_bundle}. Then we define the sheaf of particle fields via
	\begin{equation}
		\FQ \, : \quad \mathcal{T}_M \to \Gamma \left ( M, \BQ \right ) \, , \quad U \mapsto \Gamma \left ( U, B \right ) \, ,
	\end{equation}
	where \(B \subset \BQ\) is one of the subbundles from \eqnref{eqn:spacetime-matter_bundle}. More precisely, we consider the following fields:
	\begin{itemize}
		\item Lorentzian metrics \(g \in \operatorname{LorMet} \left ( M \right ) \subset \Gamma \bigl ( M, \operatorname{Sym}^2_\mathbb{R} \left ( T^* M \right ) \bigr )\)
		\item Minkowski background metric \(\eta \in \operatorname{LorMet} \left ( M \right ) \subset \Gamma \bigl ( M, \operatorname{Sym}^2_\mathbb{R} \left ( T^* M \right ) \bigr )\)
		\item Graviton fields \(\varkappa h \coloneq \left ( g - \eta \right ) \in \operatorname{Grav} \left ( M \right ) \subset \Gamma \bigl ( M, \operatorname{Sym}^2_\mathbb{R} \left ( T^* M \right ) \bigr )\), where \(\varkappa\) is the graviton coupling constant
		\item Gauge boson fields \(\imaginary \mathrm{g} A \in \operatorname{Conn} \left ( M, \mathfrak{g} \right ) \subset \Omega^1 \left ( M, \mathfrak{g} \right )\), where \(\imaginary \coloneq \sqrt{-1}\) is the imaginary unit and \(\mathrm{g}\) is the gauge boson coupling constant
		\item Graviton-ghost fields \(C \in \sctnbig{M, T^* [1,0] M}\)
		\item Graviton-antighost fields \(\overline{C} \in \sctnbig{M, T [-1,0] M}\)
		\item Graviton-Lautrup--Nakanishi auxiliary fields \(B \in \mathfrak{X} \left ( M \right )\)
		\item Gauge ghost fields \(c \in \sctnbig{M, M \times \mathfrak{g} [0,1]}\)
		\item Gauge antighost fields \(\overline{c} \in \sctnbig{M, M \times \mathfrak{g}^* [0,-1]}\)
		\item Gauge Lautrup--Nakanishi auxiliary fields \(b \in \sctnbig{M, M \times \mathfrak{g}^*}\)
	\end{itemize}
	Specifically, given a metric \(g_{\mu \nu}\) and the Minkowski background metric \(\eta_{\mu \nu}\), the graviton field \(h_{\mu \nu}\) is then defined as their difference, rescaled by the graviton coupling constant \(\varkappa \coloneq \sqrt{\kappa}\), with \(\kappa \coloneq 8 \pi G\) the Einstein constant and \(G\) the Newton constant:
	\begin{equation} \label{eqn:metric_decomposition}
		h_{\mu \nu} \coloneq \frac{1}{\varkappa} \left ( g_{\mu \nu} - \eta_{\mu \nu} \right ) \iff g_{\mu \nu} \equiv \eta_{\mu \nu} + \varkappa h_{\mu \nu} \, .
	\end{equation}
	Thus, the graviton field \(h_{\mu \nu}\) is given as a rescaled, symmetric \((0,2)\)-tensor field, i.e.\ a section \(\varkappa h \in \Gamma \bigl ( M, \operatorname{Sym}^2_\mathbb{R} \left ( T^* M \right ) \bigr )\). We remark that the Lorentz indices of the graviton field, the graviton-ghost and -antighost as well as the  corresponding Lautrup--Nakanishi auxiliary field are raised and lowered via the Minkowski background metric \(\eta_{\mu \nu}\) and its inverse \(\eta^{\mu \nu}\). Contrary, the Lorentz indices of the gauge field and all other particle fields are raised and lowered via the metric \(g_{\mu \nu}\) and its inverse \(g^{\mu \nu}\). Finally, the color indices of the gauge related fields are raised and lowered via the color metric \(\delta_{a b}\) and its inverse \(\delta^{a b}\).
\end{defn}

\enter

\begin{defn}[Diffeomorphism (anti-)BRST operator] \label{defn:diffeomorphism_brst_operator}
	We define the diffeomorphism BRST operator \(P\) as the following odd vector field on the spacetime-matter bundle with graviton-ghost degree 1:
	\begin{equation}
	\begin{split}
		P & \coloneq \frac{1}{\zeta} \left ( \nabla^{TM}_\mu C_\nu + \nabla^{TM}_\nu C_\mu \right ) \frac{\partial}{\partial h_{\mu \nu}} + \varkappa C^\rho \left ( \partial_\rho C^\sigma \right ) \frac{\partial}{\partial C^\sigma} + \frac{1}{\zeta} B^\sigma \frac{\partial}{\partial \overline{C}^\sigma} \\
		& \phantom{\coloneq} + \varkappa \sum_{\varphi \in \FQ} \left ( \Lie_C \varphi \right ) \frac{\partial}{\partial \varphi}
	\end{split}
	\end{equation}
	Equivalently, its action on fundamental particle fields is given as follows:
	{\allowdisplaybreaks
	\begin{subequations}
	\begin{align}
		\begin{split}
		P h_{\mu \nu} & \coloneq \frac{1}{\zeta} \left ( \nabla^{TM}_\mu C_\nu + \nabla^{TM}_\nu C_\mu \right ) \\ & \phantom{:} \equiv \frac{1}{\zeta} \left ( C^\rho \left ( \partial_\rho g_{\mu \nu} \right ) + \left ( \partial_\mu C^\rho \right ) g_{\rho \nu} + \left ( \partial_\nu C^\rho \right ) g_{\mu \rho} \right )
		\end{split} \\
		\begin{split}
		P g_{\mu \nu} & \coloneq \varkappa \left ( \nabla^{TM}_\mu C_\nu + \nabla^{TM}_\nu C_\mu \right ) \\ & \phantom{:} \equiv \varkappa \left ( C^\rho \left ( \partial_\rho g_{\mu \nu} \right ) + \left ( \partial_\mu C^\rho \right ) g_{\rho \nu} + \left ( \partial_\nu C^\rho \right ) g_{\mu \rho} \right )
		\end{split} \\
		P C^\rho & \coloneq \varkappa C^\sigma \left ( \partial_\sigma C^\rho \right ) \\
		P \overline{C}^\rho & \coloneq \frac{1}{\zeta} B^\rho \\
		P B^\rho & \coloneq 0 \\
		P \eta_{\mu \nu} & \coloneq 0 \\
		P \varphi & \coloneq \varkappa \left ( \Lie_C \varphi \right )
	\end{align}
	\end{subequations}
	}%
	Here, \(\Lie_C\) denotes the Lie derivative with respect to the graviton-ghost, \(\varphi\) any other particle field and \(\FQ\) the set of all such fields. Additionally, we define the diffeomorphism anti-BRST operator \(\overline{P}\) as the following odd vector field on the spacetime-matter bundle with graviton-ghost degree -1:
	\begin{subequations}
	\begin{align}
		\overline{P} & \coloneq \eval{P}_{C \rightsquigarrow \overline{C}}
		\intertext{together with the following additional changes}
		\overline{P} C^\rho & \coloneq - \frac{1}{\zeta} B^\rho + \varkappa \left ( \overline{C}^\sigma \left ( \partial_\sigma C^\rho \right ) - \bigl ( \partial_\sigma \overline{C}^\rho \bigr ) C^\sigma \right ) \\
		\overline{P} \overline{C}^\rho & \coloneq \varkappa \overline{C}^\sigma \bigl ( \partial_\sigma \overline{C}^\rho \bigr ) \\
		\overline{P} B^\rho & \coloneq \varkappa \left ( \overline{C}^\sigma \left ( \partial_\sigma B^\rho \right ) - \bigl ( \partial_\sigma \overline{C}^\rho \bigr ) B^\sigma \right )
	\end{align}
	\end{subequations}
	We remark the characteristic identities \(\commutatorbig{P}{P} = \commutatorbig{P}{\overline{P}} = \commutatorbig{\overline{P}}{\overline{P}} = 0\).
\end{defn}

\enter

\begin{defn}[Gauge (anti-)BRST operator] \label{defn:gauge_brst_operator}
	We define the gauge BRST operator \(Q\) as the following odd vector field on the spacetime-matter bundle with gauge ghost degree 1:
	\begin{equation}
	\begin{split}
		Q & \coloneq \left ( \frac{1}{\xi} \partial_\mu c^a + \mathrm{g} \tensor{f}{^a _b _c} c^b A_\mu^c \right ) \frac{\partial}{\partial A_\mu^a} + \frac{\mathrm{g}}{2} \tensor{f}{^a _b _c} c^b c^c \frac{\partial}{\partial c^a} + \frac{1}{\xi} b^a \frac{\partial}{\partial \overline{c}^a} \\ & \phantom{\coloneq} + \varkappa \sum_{\varphi \in \FQ} \left ( \lie_c \varphi \right ) \frac{\partial}{\partial \varphi}
	\end{split}
	\end{equation}
	Equivalently, its action on fundamental particle fields is given as follows:
	{\allowdisplaybreaks
	\begin{subequations}
	\begin{align}
		Q A_\mu^a & \coloneq \frac{1}{\xi} \partial_\mu c^a + \mathrm{g} \tensor{f}{^a _b _c} c^b A_\mu^c \\
		Q c^a & \coloneq \frac{\mathrm{g}}{2} \tensor{f}{^a _b _c} c^b c^c \\
		Q \overline{c}^a & \coloneq \frac{1}{\xi} b^a \\
		Q b^a & \coloneq 0 \\
		Q \delta_{ab} & \coloneq 0 \\
		Q \varphi & \coloneq \mathrm{g} \left ( \lie_c \varphi \right )
	\end{align}
	\end{subequations}
	}%
	Here, \(\lie_c\) denotes the Lie derivative with respect to the gauge ghost, \(\varphi\) any other particle field and \(\FQ\) the set of all such fields. Additionally, we define the gauge anti-BRST operator \(\overline{Q}\) as the following odd vector field on the spacetime-matter bundle with gauge ghost degree -1:
	\begin{subequations}
	\begin{align}
		\overline{Q} & \coloneq \eval{Q}_{c \rightsquigarrow \overline{c}}
		\intertext{together with the following additional changes}
		\overline{Q} c^a & \coloneq - \frac{1}{\xi} b^a + \mathrm{g} \tensor{f}{^a _b _c} \overline{c}^b c^c \\
		\overline{Q} \overline{c}^a & \coloneq \frac{\mathrm{g}}{2} \tensor{f}{^a _b _c} \overline{c}^b \overline{c}^c \\
		\overline{Q} b^a & \coloneq \mathrm{g} \tensor{f}{^a _b _c} \overline{c}^b b^c
	\end{align}
	\end{subequations}
	We remark the characteristic identities \(\commutatorbig{Q}{Q} = \commutatorbig{Q}{\overline{Q}} = \commutatorbig{\overline{Q}}{\overline{Q}} = 0\).
\end{defn}

\enter

\begin{defn}[Total (anti-)BRST operator\footnote{The \emph{total BRST operator} and \emph{total anti-BRST operator} have been introduced and studied in \cite[Section 5]{Prinz_5}.}] \label{defn:total_brst_operator}
	Let \(P\) and \(Q\) be the diffeomorphism and gauge BRST operators from \defnsaref{defn:diffeomorphism_brst_operator}{defn:gauge_brst_operator}, we call their sum
	\begin{equation}
		D \coloneq P + Q
	\end{equation}
	the \emph{total BRST operator}. In addition, let \(\overline{P}\) and \(\overline{Q}\) be the diffeomorphism and gauge anti-BRST operators from \defnsaref{defn:diffeomorphism_brst_operator}{defn:gauge_brst_operator}, we call their sum
	\begin{equation}
		\overline{D} \coloneq \overline{P} + \overline{Q}
	\end{equation}
	the \emph{total anti-BRST operator}. We remark that both operators are indeed anticommuting differentials due to \cite[Theorem 5.1 and Corollary 5.2]{Prinz_5} and thus satisfy the characteristic identities \(\commutatorbig{D}{D} = \commutatorbig{D}{\overline{D}} = \commutatorbig{\overline{D}}{\overline{D}} = 0\).
\end{defn}

\enter

\begin{defn}[Super-BRST operators] \label{defn:super-brst-operator}
	Given the BRST operators \(D, P, Q\) and their corresponding anti-BRST operators \(\overline{D}, \overline{P}, \overline{Q}\), then we define the respective \emph{super-BRST operators} as follows:
	\begin{align}
		\mathcal{D} & \coloneq D \circ \overline{D} \, , \\
		\mathcal{P} & \coloneq P \circ \overline{P}
		\intertext{and}
		\mathcal{Q} & \coloneq Q \circ \overline{Q}
	\end{align}
	In particular, they are even vector fields on the spacetime-matter bundle with ghost degrees 0.
\end{defn}

\enter

\begin{rem}
	The super-BRST operators from \defnref{defn:super-brst-operator} are also nilpotent, i.e.\ satisfy
	\begin{equation}
		\mathcal{D}^2 = \mathcal{P}^2 = \mathcal{Q}^2 = 0 \, ,
	\end{equation}
	due to the anticommutativity and nilpotency of the BRST operators with their respective anti-BRST operators, cf.\ \cite[Corollaries 3.4, 4.4 and 5.2]{Prinz_5}.
\end{rem}

\enter

\begin{defn}[Ghost conjugation, anti-Hermitian auxiliary field] \label{defn:ghost-conjugation}
	We introduce the following three Hermitian involutions on the space of particle fields \(\FQ\): First, the graviton-ghost conjugation \(\dagger_C\) and the gauge ghost conjugation \(\dagger_c\) via (\(\varphi\) denotes again any other particle field):
	{\allowdisplaybreaks
	\begin{subequations}
	\begin{align}
		& \bigl ( C^\rho \bigr )^{\dagger_C} \coloneq \overline{C}^\rho && \bigl ( C^\rho \bigr )^{\dagger_c} \coloneq C^\rho \\
		& \bigl ( \overline{C}^\rho \bigr )^{\dagger_C} \coloneq C^\rho && \bigl ( \overline{C}^\rho \bigr )^{\dagger_c} \coloneq \overline{C}^\rho \\
		& \bigl ( B^\rho \bigr )^{\dagger_C} \coloneq - B^\rho - \varkappa \zeta \left ( \overline{C}^\sigma \bigl ( \partial_\sigma C^\rho \bigr ) - \bigl ( \partial_\sigma \overline{C}^\rho \bigr ) C^\sigma \right ) && \bigl ( B^\rho \bigr )^{\dagger_c} \coloneq B^\rho \\
		& \bigl ( {B^\prime}^\rho \bigr )^{\dagger_C} \coloneq - {B^\prime}^\rho && \bigl ( {B^\prime}^\rho \bigr )^{\dagger_c} \coloneq {B^\prime}^\rho \\
		& \bigl ( c^a \bigr )^{\dagger_C} \coloneq c^a && \bigl ( c^a \bigr )^{\dagger_c} \coloneq \overline{c}^a \\
		& \bigl ( \overline{c}^a \bigr )^{\dagger_C} \coloneq \overline{c}^a && \bigl ( \overline{c}^a \bigr )^{\dagger_c} \coloneq c^a \\
		& \bigl ( b^a \bigr )^{\dagger_C} \coloneq b^a && \bigl ( b^a \bigr )^{\dagger_c} \coloneq - b^a - \mathrm{g} \xi \tensor{f}{^a _b _c} \overline{c}^b c^c \\
		& \bigl ( {b^\prime}^a \bigr )^{\dagger_C} \coloneq {b^\prime}^a && \bigl ( {b^\prime}^a \bigr )^{\dagger_c} \coloneq - {b^\prime}^a \\
		& \bigl ( \partial_\mu \bigr )^{\dagger_C} \coloneq - \partial_\mu && \bigl ( \partial_\mu \bigr )^{\dagger_c} \coloneq - \partial_\mu \\
		& \bigl ( \tensor{\Gamma}{^\rho _\mu _\nu} \bigr )^{\dagger_C} \coloneq - \tensor{\Gamma}{^\rho _\mu _\nu} && \bigl ( \tensor{\Gamma}{^\rho _\mu _\nu} \bigr )^{\dagger_c} \coloneq - \tensor{\Gamma}{^\rho _\mu _\nu} \\
		& \bigl ( \imaginary \tensor{f}{^a _b _c} \bigr )^{\dagger_C} \coloneq - \imaginary \tensor{f}{^a _b _c} && \bigl ( \imaginary \tensor{f}{^a _b _c} \bigr )^{\dagger_c} \coloneq - \imaginary \tensor{f}{^a _b _c} \\
		& \bigl ( \varphi \bigr )^{\dagger_C} \coloneq \varphi && \bigl ( \varphi \bigr )^{\dagger_c} \coloneq \varphi
	\end{align}
	\end{subequations}
	}%
	Here, \({B^\prime}^\rho\) and \({b^\prime}^a\) are the shifted anti-Hermitian Lautrup--Nakanishi auxiliary fields, given as follows:
	\begin{align}
		{B^\prime}^\rho & \coloneq B^\rho - \frac{\varkappa \zeta}{2} \left ( \overline{C}^\sigma \bigl ( \partial_\sigma C^\rho \bigr ) - \bigl ( \partial_\sigma \overline{C}^\rho \bigr ) C^\sigma \right )
		\intertext{and}
		{b^\prime}^a & \coloneq b^a - \frac{\mathrm{g} \xi}{2} \tensor{f}{^a _b _c} \overline{c}^b c^c \, .
	\end{align}
	And then, finally, we introduce the total ghost conjugation \(\dagger\) as follows:
	{\allowdisplaybreaks
	\begin{subequations}
	\begin{align}
		& \bigl ( C^\rho \bigr )^\dagger \coloneq \overline{C}^\rho \\
		& \bigl ( \overline{C}^\rho \bigr )^\dagger \coloneq C^\rho \\
		& \bigl ( B^\rho \bigr )^\dagger \coloneq - B^\rho - \varkappa \zeta \left ( \overline{C}^\sigma \bigl ( \partial_\sigma C^\rho \bigr ) - \bigl ( \partial_\sigma \overline{C}^\rho \bigr ) C^\sigma \right ) \\
		& \bigl ( {B^\prime}^\rho \bigr )^\dagger \coloneq - {B^\prime}^\rho \\
		& \bigl ( c^a \bigr )^\dagger \coloneq \overline{c}^a \\
		& \bigl ( \overline{c}^a \bigr )^\dagger \coloneq c^a \\
		& \bigl ( b^a \bigr )^\dagger \coloneq - b^a - \mathrm{g} \xi \tensor{f}{^a _b _c} \overline{c}^b c^c \\
		& \bigl ( {b^\prime}^a \bigr )^\dagger \coloneq - {b^\prime}^a \\
		& \bigl ( \partial_\mu \bigr )^\dagger \coloneq - \partial_\mu \\
		& \bigl ( \tensor{\Gamma}{^\rho _\mu _\nu} \bigr )^\dagger \coloneq - \tensor{\Gamma}{^\rho _\mu _\nu} \\
		& \bigl ( \imaginary \tensor{f}{^a _b _c} \bigr )^\dagger \coloneq - \imaginary \tensor{f}{^a _b _c} \\
		& \bigl ( \varphi \bigr )^\dagger \coloneq \varphi
	\end{align}
	\end{subequations}
	}%
	In particular, the total ghost conjugation inverts simultaneously graviton-ghosts and gauge ghosts.
\end{defn}

\enter

\begin{rem}
	The super-BRST operators are anti-Hermitian with respect to their associated ghost conjugation:\footnote{In addition, all conjugated BRST operators act to the left.}
	\begin{subequations}
	\begin{align}
		\mathcal{P}^{\dagger_C} & \equiv - \mathcal{P} \\
		\mathcal{Q}^{\dagger_c} & \equiv - \mathcal{Q} \\
		\mathcal{D}^\dagger & \equiv - \mathcal{D}
	\end{align}
	\end{subequations}
	 In addition, we remark that the anti-BRST operators are related to their corresponding BRST operators via ghost-conjugation, cf.\ \cite[Lemma 5.7]{Prinz_5}:
	 \begin{subequations}
	 \begin{align}
		\overline{P} & \equiv P^{\dagger_C} \\
		\overline{Q} & \equiv Q^{\dagger_c} \\
		\overline{D} & \equiv D^{\dagger}
	\end{align}
	\end{subequations}
\end{rem}

\section{The case of (effective) Quantum General Relativity} \label{sec:case-qgr}

We calculate the symmetric gauge fixing and ghost Lagrange density for (effective) Quantum General Relativity with a linearized de Donder gauge fixing condition in \propref{prop:symmetric-ghost-qgr}. Then we relate the symmetric setting to the Faddeev--Popov and opposed Faddeev--Popov constructions in \thmref{thm:symmetric-ghost-homotopy-qgr}. Finally, we discuss further possible choices for the gauge fixing boson in \remref{rem:qgr-gauge-fixing-boson}.

\enter

\begin{prop} \label{prop:symmetric-ghost-qgr}
	The symmetric gauge fixing and ghost Lagrange density for (effective) Quantum General Relativity reads
	\begin{equation}
	\begin{split}
		\mathcal{L}_\textup{GR-GF-Sym-Ghost} & \coloneq \frac{1}{2 \zeta} \left ( - \frac{1}{2 \varkappa^2} \eta^{\mu \nu} \deDonder^{(1)}_\mu \deDonder^{(1)}_\nu + \eta^{\mu \nu} \bigl ( \partial_\mu \overline{C}^\rho \bigr ) \bigl ( \partial_\nu C_\rho \bigr ) \right ) \dif V_\eta \\
	& \phantom{\coloneq} + \frac{1}{4} \eta^{\mu \nu} \left ( \bigl ( \partial_\rho \overline{C}^\rho \bigr ) \bigl ( \Gamma_{\sigma \mu \nu} C^\sigma \bigr ) - 2 \bigl ( \partial_\mu \overline{C}^\rho \bigr ) \bigl ( \Gamma_{\sigma \rho \nu} C^\sigma \bigr ) \right ) \dif V_\eta \\
	& \phantom{\coloneq} + \frac{1}{4} \eta^{\mu \nu} \left ( \bigl ( \Gamma_{\sigma \mu \nu} \overline{C}^\sigma \bigr ) \bigl ( \partial_\rho C^\rho \bigr ) - 2 \bigl ( \Gamma_{\sigma \rho \nu} \overline{C}^\sigma \bigr ) \bigl ( \partial_\mu C^\rho \bigr ) \right ) \dif V_\eta \\
	& \phantom{\coloneq} + \frac{\varkappa^2 \zeta}{8} \eta_{\mu \nu} \left ( \overline{C}^\rho \bigl ( \partial_\rho \overline{C}^\mu \bigr ) \right ) \left ( C^\sigma \bigl ( \partial_\sigma C^\nu \bigr ) \right ) \dif V_\eta
	\end{split}
	\end{equation}
	with the linearized de Donder gauge fixing functional \(\deDonder^{(1)}_\mu \coloneq \eta^{\rho \sigma} \Gamma_{\mu \rho \sigma}\). It can be obtained from the gauge fixing boson
	\begin{equation} \label{eqn:qgr-gauge-fixing-boson}
		F \coloneq - \frac{\zeta}{4} \left ( \frac{1}{\varkappa} \eta^{\mu \nu} h_{\mu \nu} - \overline{C}^\rho C_\rho \right ) \dif V_\eta
	\end{equation}
	via \(\mathcal{P} F\), where \(\mathcal{P}\) is the diffeomorphism super-BRST operator.
\end{prop}

\begin{proof}
	The claimed statement results directly from the following calculations:\footnote{We emphasize the relation of the symmetric gauge fixing fermion \(\overline{P} F\) to the Faddeev--Popov gauge fixing fermion \(\stigma^{(1)}\), given in \cite[Equation (44)]{Prinz_5} using the same conventions:
	\begin{equation}
		\overline{P} F \equiv \stigma^{(1)} - \frac{\varkappa \zeta}{4} \overline{C}^\rho \bigl ( \partial_\rho \overline{C}^\sigma \bigr ) C_\sigma
	\end{equation}
	In addition, to achieve the symmetric setup, the Lautrup--Nakanishi auxiliary field needs to be shifted to become anti-Hermitian, cf.\ \eqnref{eqn:shifted_ln_field_gr}. This is implicitly included in the homotopy gauge fixing fermion of \eqnref{eqn:qgr-homotopy-gff}.}
	\begin{subequations}
	\begin{align}
		\begin{split}
			\overline{P} F & = - \frac{1}{4 \varkappa} \eta^{\mu \nu} \left ( \overline{C}^\rho \bigl ( \partial_\rho g_{\mu \nu} \bigr ) + \bigl ( \partial_\mu \overline{C}^\rho \bigr ) g_{\rho \nu} + \bigl ( \partial_\nu \overline{C}^\rho \bigr ) g_{\mu \rho} \right ) \dif V_\eta \\
			& \phantom{=} + \left ( \frac{\varkappa \zeta}{4} \left ( \overline{C}^\rho \bigl ( \partial_\rho \overline{C}^\sigma \bigr ) C_\sigma - \overline{C}^\rho \overline{C}^\sigma \bigl ( \partial_\sigma C_\rho \bigr ) + \overline{C}^\rho \bigl ( \partial_\sigma \overline{C}_\rho \bigr ) C^\sigma \right ) + \frac{1}{4} \overline{C}^\rho B_\rho \right ) \dif V_\eta \\
			& \simeq_\text{PI \& IDs} - \frac{1}{4 \varkappa} \eta^{\mu \nu} \left ( \overline{C}^\rho \bigl ( \partial_\rho g_{\mu \nu} \bigr ) - \overline{C}^\rho \bigl ( \partial_\mu g_{\rho \nu} \bigr ) - \overline{C}^\rho \bigl ( \partial_\nu g_{\mu \rho} \bigr ) \right ) \dif V_\eta \\
			& \phantom{=} + \left ( - \frac{\varkappa \zeta}{4} \overline{C}^\rho \bigl ( \partial_\rho \overline{C}^\sigma \bigr ) C_\sigma + \frac{1}{4} \overline{C}^\rho B_\rho \right ) \dif V_\eta \\
			& = \overline{C}^\rho \left ( \frac{1}{2 \varkappa} \deDonder_\rho - \frac{\varkappa \zeta}{4} \bigl ( \partial_\rho \overline{C}^\sigma \bigr ) C_\sigma + \frac{1}{4} B_\rho \right ) \dif V_\eta \, , \label{eqn:anti-p_f}
		\end{split}
		\intertext{where \(\simeq_\text{PI \& IDs}\) denotes equality modulo partial integration and Lie algebra identities (using that the Minkowski background metric \(\eta_{\mu \nu}\) is by definition invariant), and}
		\begin{split}
			\bigl ( P \circ \overline{P} \bigr ) F & \simeq_\text{PI \& IDs} \left ( \frac{1}{4 \zeta} B^\rho B_\rho + \frac{1}{2 \varkappa \zeta} B^\rho \deDonder_\rho + \frac{1}{2 \zeta} \eta^{\mu \nu} \bigl ( \partial_\mu \overline{C}^\rho \bigr ) \bigl ( \partial_\nu C_\rho \bigr ) \right ) \dif V_\eta \\
			& \phantom{\simeq} + \eta^{\mu \nu} \left ( \frac{1}{2} \bigl ( \partial_\rho \overline{C}^\rho \bigr ) \bigl ( \Gamma_{\sigma \mu \nu} C^\sigma \bigr ) - \bigl ( \partial_\mu \overline{C}^\rho \bigr ) \bigl ( \Gamma_{\sigma \rho \nu} C^\sigma \bigr ) \right ) \dif V_\eta \\
			& \phantom{\simeq} + \left ( - \frac{\varkappa}{4} B^\rho \bigl ( \partial_\rho \overline{C}^\sigma \bigr ) C_\sigma + \frac{\varkappa \zeta}{4} \overline{C}^\rho \bigl ( \partial_\rho B^\sigma \bigr ) C_\sigma \right ) \dif V_\eta \\
			& \phantom{\simeq} + \left ( \frac{\varkappa^2 \zeta}{16} \overline{C}^\mu \bigl ( \partial_\mu \overline{C}^\rho \bigr ) C^\nu \bigl ( \partial_\nu C_\rho \bigr ) \right ) \dif V_\eta \, , \label{eqn:p_anti-p_f}
		\end{split}
		\intertext{together with the replacement of the Lautrup--Nakanishi auxiliary field with its anti-Hermitian shift}
		B^\rho & \equiv {B^\prime}^\rho + \frac{\varkappa \zeta}{2} \left ( \overline{C}^\sigma \bigl ( \partial_\sigma C^\rho \bigr ) - \bigl ( \partial_\sigma \overline{C}^\rho \bigr ) C^\sigma \right ) \label{eqn:shifted_ln_field_gr}
		\intertext{and then finally eliminating the shifted auxiliary field \({B^\prime}^\rho\) by inserting its equation of motion}
		\operatorname{EoM} \bigl ( B^\prime_\rho \bigr ) & = \frac{1}{\varkappa} \deDonder_\rho \, ,
	\intertext{which are obtained as usual via an Euler--Lagrange variation of \eqnref{eqn:p_anti-p_f}, i.e.\ by solving}
	0 & \overset{!}{=} \left ( \left ( \frac{\partial}{\partial B^\prime_\rho} \right ) - \partial_\mu \left ( \frac{\partial}{\partial \, \bigl ( \partial_\mu B^\prime_\rho \bigr )} \right ) \right ) \mathcal{P} F \, ,
	\end{align}
	\end{subequations}
	where the second term vanishes identically, as \(B^\prime_\rho\) is a Lagrange multiplier and thus has no kinetic term after a suitable partial integration.
\end{proof}

\enter

\begin{thm} \label{thm:symmetric-ghost-homotopy-qgr}
	We obtain the following homotopy in \(\varsigma \in [0, 1]\) between the Faddeev--Popov construction \(\varsigma = 0\), the symmetric setting \(\varsigma = \textfrac{1}{2}\) and the opposed Faddeev--Popov construction \(\varsigma = 1\):
	\begin{equation} \label{eqn:qgr-gf-ghost-homotopy}
	\begin{split}
		\mathcal{L}_\textup{GR-GF-Ghost} \left ( \varsigma \right ) & \coloneq \frac{1}{2 \zeta} \left ( - \frac{1}{2 \varkappa^2} \eta^{\mu \nu} \deDonder^{(1)}_\mu \deDonder^{(1)}_\nu + \eta^{\mu \nu} \bigl ( \partial_\mu \overline{C}^\rho \bigr ) \bigl ( \partial_\nu C_\rho \bigr ) \right ) \dif V_\eta \\
		& \phantom{\coloneq} + \frac{\left ( 1 - \varsigma \right )}{2} \eta^{\mu \nu} \left ( \bigl ( \partial_\rho \overline{C}^\rho \bigr ) \bigl ( \Gamma_{\sigma \mu \nu} C^\sigma \bigr ) - 2 \bigl ( \partial_\mu \overline{C}^\rho \bigr ) \bigl ( \Gamma_{\sigma \rho \nu} C^\sigma \bigr ) \right ) \dif V_\eta \\
		& \phantom{\coloneq} + \frac{\varsigma}{2} \eta^{\mu \nu} \left ( \bigl ( \Gamma_{\sigma \mu \nu} \overline{C}^\sigma \bigr ) \bigl ( \partial_\rho C^\rho \bigr ) - 2 \bigl ( \Gamma_{\sigma \rho \nu} \overline{C}^\sigma \bigr ) \bigl ( \partial_\mu C^\rho \bigr ) \right ) \dif V_\eta \\
		& \phantom{\coloneq} + \frac{\varkappa^2 \zeta \varsigma \left ( 1 - \varsigma \right )}{2} \eta_{\mu \nu} \left ( \overline{C}^\rho \bigl ( \partial_\rho \overline{C}^\mu \bigr ) \right ) \left ( C^\sigma \bigl ( \partial_\sigma C^\nu \bigr ) \right ) \dif V_\eta
	\end{split}
	\end{equation}
	We call \(\varsigma\) the graviton-ghost parameter. In particular, this unifies the Faddeev--Popov construction, cf.\ \cite[Corollary 3.7]{Prinz_5}, with the symmetric construction of \propref{prop:symmetric-ghost-qgr}.\footnote{We remark that the relative minus sign of the ghost Lagrange density in comparison with \cite[Corollary 3.7]{Prinz_5} is due to a partial integration and emphasize that the used conventions are indeed the same.} Specifically, it is generated using the graviton homotopy gauge fixing fermion
	\begin{equation}
		\stigma \left ( \varsigma \right ) \coloneq \left ( - \frac{\zeta}{4 \varkappa} \eta^{\mu \nu} \overline{P} \bigl ( h_{\mu \nu} \bigr ) + \frac{\zeta \varsigma}{8} \overline{P} \bigl ( \overline{C}^\rho C_\rho \bigr ) + \left ( \frac{\varsigma}{4} - \frac{1}{2} \right ) \overline{C}^\rho B_\rho \right ) \dif V_\eta \label{eqn:qgr-homotopy-gff}
	\end{equation}
	via \(P \stigma \left ( \varsigma \right )\).
\end{thm}

\begin{proof}
	This can be shown analogously to \propref{prop:symmetric-ghost-qgr}.
\end{proof}

\enter

\begin{rem} \label{rem:qgr-gauge-fixing-boson}
	In addition to the linear gauge fixing boson of \eqnref{eqn:qgr-gauge-fixing-boson}, it is also possible to use a quadratic version
	\begin{align}
		F^{(2)} & \coloneq - \frac{\zeta}{4} \left ( \frac{1}{\varkappa} h^{\mu \nu} h_{\mu \nu} - \overline{C}^\rho C_\rho \right ) \dif V_\eta \label{eqn:qgr-gauge-fixing-boson-quadratic}
		\intertext{and the following infinite series}
		\begin{split} \label{eqn:qgr-gauge-fixing-boson-infinity}
			F^{(\infty)} & \coloneq - \frac{\zeta}{4} \left ( \frac{1}{\varkappa} g^{\mu \nu} h_{\mu \nu} - \overline{C}^\rho C_\rho \right ) \dif V_\eta \\
			& \phantom{:} \simeq \frac{\zeta}{4} \left ( \frac{1}{\varkappa^2} g^{\mu \nu} \eta_{\mu \nu} + \overline{C}^\rho C_\rho \right ) \dif V_\eta \, ,
		\end{split}
	\end{align}
	where we have used the identity \(g^{\mu \nu} h_{\mu \nu} \equiv \textfrac{g^{\mu \nu} \left ( g_{\mu \nu} - \eta_{\mu \nu} \right )}{\varkappa} \equiv \textfrac{\left ( d - g^{\mu \nu} \eta_{\mu \nu} \right )}{\varkappa}\) with \(d\) the dimension of spacetime, and then dropped the constant term, as it would not contribute on the level of the Lagrange density. However, the reason why we are using the linear variant of \eqnref{eqn:qgr-gauge-fixing-boson} is due to the fact that it reproduces the linearized de Donder gauge fixing condition. The gauge fixing bosons of \eqnsaref{eqn:qgr-gauge-fixing-boson-quadratic}{eqn:qgr-gauge-fixing-boson-infinity} correspond to different gauge fixing conditions and produce different ghost Lagrange densities (even on the propagator level). While we prefer the linearized variant of \propref{prop:symmetric-ghost-qgr} and \thmref{thm:symmetric-ghost-homotopy-qgr}, as it connects nicely to the corresponding Faddeev--Popov variant presented in \cite[Corollary 3.7]{Prinz_5} and used in \cite{Prinz_2,Prinz_4}, it might also be worthwhile to study the other variants in future work. Specifically, this differs from the situation in Quantum Yang--Mills theory, where the potential term is necessarily quadratic, cf.\ \eqnref{eqn:qym-gauge-fixing-boson}.
\end{rem}

\section{The case of covariant Quantum Yang--Mills theory} \label{sec:case-qym}

We calculate the symmetric gauge fixing and ghost Lagrange density for covariant Quantum Yang--Mills theory with a covariant Lorenz gauge fixing condition in \propref{prop:symmetric-ghost-qym}. Then we relate the symmetric setting to the Faddeev--Popov and opposed Faddeev--Popov constructions in \thmref{thm:symmetric-ghost-homotopy-qym}. Finally, we exemplify the construction and in particular the achieved curved-spacetime generalization with respect to \cite{Baulieu_Thierry-Mieg} in \exsaref{exmp:qym-su-two}{exmp:qym-schwarzschild}.

\enter

\begin{prop} \label{prop:symmetric-ghost-qym}
	The symmetric gauge fixing and ghost Lagrange density for Quantum Yang--Mills theory reads
	\begin{equation}
		\begin{split}
			\mathcal{L}_\textup{YM-GF-Sym-Ghost} & \coloneq \frac{1}{\xi} \left ( - \frac{1}{2 \mathrm{g}^2} \delta_{ab} \Lorenz^a \Lorenz^b + g^{\mu \nu} \left ( \partial_\mu \overline{c}_a \right ) \left ( \partial_\nu c^a \right ) \right ) \dif V_g \\
			& \phantom{\coloneq} + \frac{\mathrm{g}}{2} g^{\mu \nu} \tensor{f}{^a _b _c} \left ( \bigl ( \partial_\mu \overline{c}_a \bigr ) \bigl ( c^b A^c_\nu \bigr ) + \bigl ( \overline{c}_a A^b_\nu \bigr ) \bigl ( \partial_\mu c^c \bigr ) \right ) \dif V_g \\
			& \phantom{\coloneq} + \frac{\mathrm{g}^2 \xi}{16} f^{a b c} f_{ade} \overline{c}_b \overline{c}_c c^d c^e \dif V_g
		\end{split}
	\end{equation}
	with the covariant Lorenz gauge fixing functional \(\Lorenz^a \coloneq \mathrm{g} g^{\mu \nu} \bigl ( \nabla^{TM}_\mu A^a_\nu \bigr ) \equiv 0\). It can be obtained from the gauge fixing boson
	\begin{equation} \label{eqn:qym-gauge-fixing-boson}
		G \coloneq - \frac{\xi}{2} \left ( \delta_{a b} g^{\mu \nu} A^a_\mu A^b_\nu - \overline{c}_a c^a \right ) \dif V_g
	\end{equation}
	via \(\mathcal{Q} G\), where \(\mathcal{Q}\) is the gauge super-BRST operator.
\end{prop}

\begin{proof}
	The claimed statement results directly from the following calculations:\footnote{We emphasize the relation of the symmetric gauge fixing fermion \(\overline{Q} G\) to the Faddeev--Popov gauge fixing fermion \(\digamma \! \! _{\{ 1 \}}\), given in \cite[Equation (56)]{Prinz_5} using the same conventions:
	\begin{equation}
		\overline{Q} G \equiv \digamma \! \! _{\{ 1 \}} - \frac{\mathrm{g} \xi}{4} \tensor{f}{^a _b _c} \overline{c}_a \overline{c}^b c^c
	\end{equation}
	In addition, to achieve the symmetric setup, the Lautrup--Nakanishi auxiliary field needs to be shifted to become anti-Hermitian, cf.\ \eqnref{eqn:shifted_ln_field_ym}. This is implicitly included in the homotopy gauge fixing fermion of \eqnref{eqn:qym-homotopy-gff}.}
	\begin{subequations}
	\begin{align}
		\begin{split}
			\overline{Q} G & = \left ( - \delta_{a b} g^{\mu \nu} A^a_\mu \left ( \partial_\nu \overline{c}^b + \xi \mathrm{g} \tensor{f}{^b _c _d} \overline{c}^c A_\nu^d \right ) - \frac{\mathrm{g} \xi}{4} \tensor{f}{^a _b _c} \overline{c}^b \overline{c}^c c_a + \frac{1}{2} \overline{c}_a b^a \right ) \dif V_g \\
			& \simeq_\text{PI} \left ( g^{\mu \nu} \left ( \overline{c}_a \left ( \nabla^{TM}_\mu A^a_\nu \right ) + \xi \mathrm{g} \tensor{f}{^a _b _c} \overline{c}_a A^b_\mu A_\nu^c \right ) - \frac{\mathrm{g} \xi}{4} \tensor{f}{^a _b _c} \overline{c}^b \overline{c}^c c_a + \frac{1}{2} \overline{c}_a b^a \right ) \dif V_g \\
			& = \overline{c}_a \left ( \frac{1}{\mathrm{g}} \Lorenz^a + \underbrace{\xi \mathrm{g} g^{\mu \nu} \tensor{f}{^a _b _c} A^b_\mu A_\nu^c}_{= 0} - \frac{\mathrm{g} \xi}{4} \tensor{f}{^a _b _c} \overline{c}^b c^c + \frac{1}{2} b^a \right ) \dif V_g \, , \label{eqn:anti-q_g}
		\end{split}
		\intertext{where \(\simeq_\text{PI}\) denotes equality modulo partial integration, and}
		\begin{split}
			\bigl ( Q \circ \overline{Q} \bigr ) G & \simeq_\text{PI} \left ( \frac{1}{2 \xi} b_a b^a  + \frac{1}{\mathrm{g} \xi} b_a \Lorenz^a + \frac{1}{\xi} g^{\mu \nu} \bigl ( \partial_\mu \overline{c}_a \bigr ) \bigl ( \partial_\nu c^a \bigr ) \right ) \dif V_g \\
			& \phantom{\simeq} + \left ( \mathrm{g} \tensor{f}{^a _b _c} \bigl ( \partial_\mu \overline{c}_a \bigr ) c^b A^c_\nu - \frac{\mathrm{g}}{2} \tensor{f}{^a _b _c} b^b \overline{c}^c c_a + \frac{\mathrm{g}^2 \xi}{8} \tensor{f}{^a _b _c} \tensor{f}{_a _d _e} \overline{c}^b \overline{c}^c c^d c^e \right ) \dif V_g \, , \label{eqn:q_anti-q_g}
		\end{split}
		\intertext{together with the replacement of the Lautrup--Nakanishi auxiliary field with its anti-Hermitian shift}
		b^a & \equiv {b^\prime}^a + \frac{\mathrm{g} \xi}{2} \tensor{f}{^a _b _c} \overline{c}^b c^c \label{eqn:shifted_ln_field_ym}
		\intertext{and then finally eliminating the shifted auxiliary field \({b^\prime}^a\) by inserting its equation of motion}
		\operatorname{EoM} \bigl ( {b^\prime}^a \bigr ) & = \frac{1}{\mathrm{g}} \Lorenz^a \, ,
	\intertext{which are obtained as usual via an Euler--Lagrange variation of \eqnref{eqn:q_anti-q_g}, i.e.\ by solving}
	0 & \overset{!}{=} \left ( \left ( \frac{\partial}{\partial b^\prime_a} \right ) - \partial_\mu \left ( \frac{\partial}{\partial \, \bigl ( \partial_\mu b^\prime_a \bigr )} \right ) \right ) \mathcal{Q} G \, ,
	\end{align}
	\end{subequations}
	where the second term vanishes identically, as \(b^\prime_a\) is a Lagrange multiplier and thus has no kinetic term.
\end{proof}

\enter

\begin{thm} \label{thm:symmetric-ghost-homotopy-qym}
	We obtain the following homotopy in \(\vartheta \in [0, 1]\) between the Faddeev--Popov construction \(\vartheta = 0\), the symmetric setting \(\vartheta = \textfrac{1}{2}\) and the opposed Faddeev--Popov construction \(\vartheta = 1\):
	\begin{equation} \label{eqn:ym-gf-ghost-homotopy}
	\begin{split}
		\mathcal{L}_\textup{YM-GF-Ghost} \left ( \vartheta \right ) & \coloneq \frac{1}{\xi} \left ( - \frac{1}{2 \mathrm{g}^2} \delta_{ab} \Lorenz^a \Lorenz^b + g^{\mu \nu} \left ( \partial_\mu \overline{c}_a \right ) \left ( \partial_\nu c^a \right ) \right ) \dif V_g \phantom{\bigg \vert} \\
		& \phantom{\coloneq} + \frac{\mathrm{g}}{2} g^{\mu \nu} \tensor{f}{^a _b _c} \left ( \left ( 1 - \vartheta \right ) \bigl ( \partial_\mu \overline{c}_a \bigr ) \bigl ( c^b A^c_\nu \bigr ) + \vartheta \bigl ( \overline{c}_a A^b_\nu \bigr ) \bigl ( \partial_\mu c^c \bigr ) \right ) \dif V_g \phantom{\bigg \vert} \\
		& \phantom{\coloneq} + \frac{\mathrm{g}^2 \xi \vartheta \left ( 1 - \vartheta \right )}{4} f^{a b c} f_{ade} \overline{c}_b \overline{c}_c c^d c^e \dif V_g \phantom{\bigg \vert}
	\end{split}
	\end{equation}
	We call \(\vartheta\) the gauge ghost parameter. In particular, this unifies the Faddeev--Popov construction, cf.\ \cite[Corollary 4.6]{Prinz_5}, with the symmetric construction of \propref{prop:symmetric-ghost-qym}.\footnote{We remark that the relative minus sign of the ghost Lagrange density in comparison with \cite[Corollary 4.6]{Prinz_5} is due to a partial integration and emphasize that the used conventions are indeed the same.} Specifically, it is generated using the gauge boson homotopy gauge fixing fermion
	\begin{equation}
		\digamma \! \left ( \vartheta \right ) \coloneq \left ( - \frac{\xi}{2} \delta_{ab} g^{\mu \nu} \overline{Q} \bigl ( A^a_\mu A^b_\nu \bigr ) + \frac{\xi \vartheta}{4} \overline{Q} \left ( \overline{c}_a c^a \right ) + \left ( \frac{\vartheta}{2} - 1 \right ) \overline{c}_a b^a \right ) \dif V_g \label{eqn:qym-homotopy-gff}
	\end{equation}
	via \(Q \digamma \! \left ( \vartheta \right )\).
\end{thm}

\begin{proof}
	This can be shown analogously to \propref{prop:symmetric-ghost-qym}.
\end{proof}

\enter

\begin{exmp}[\(\operatorname{SU}(2)\) gauge group] \label{exmp:qym-su-two}
	To exemplify the covariant Yang--Mills theory ghost homotopy construction, we work out the specific situation of a \(\operatorname{SU}(2)\) gauge group here and that of a Schwarzschild background in \exref{exmp:qym-schwarzschild} --- both representing the respective simplest non-trivial applications: To this end, we remind the reader that we call the sum of the classical Yang--Mills theory Lagrange density \(\mathcal{L}_\text{YM}\) with the gauge fixing and ghost Lagrange density \(\mathcal{L}_\text{YM-GF-Ghost}\) the \emph{Quantum Yang--Mills theory Lagrange density}, i.e.\
	\begin{equation}
		\mathcal{L}_\text{QYM} \coloneq \mathcal{L}_\text{YM} + \mathcal{L}_\text{YM-GF-Ghost} \, ,
	\end{equation}
	because this Lagrange density allows for a perturbative quantization.\footnote{The gauge fixing Lagrange density is needed in order to calculate the gauge boson propagator and the ghost Lagrange density is needed to obtain a transversal perturbative expansion.} We start by expanding the classical Yang--Mills theory Lagrange density, as stated in \eqnref{eqn:ym_lagrange_density_introduction}, using \(F^a_{\mu \nu} \coloneq \mathrm{g} \bigl ( \partial_\mu A^a_\nu - \partial_\nu A^a_\mu \bigr ) - \mathrm{g}^2 \tensor{f}{^a _b _c} A^b_\mu A^c_\nu\), as follows:
	\begin{equation} \label{eqn:ym_lagrange_density_expanded}
	\begin{split}
		\mathcal{L}_\text{YM} & \coloneq - \frac{1}{4 \mathrm{g}^2} \delta_{a b} g^{\mu \nu} g^{\rho \sigma} F^a_{\mu \rho} F^b_{\nu \sigma} \dif V_g \\
		& \phantom{:} \equiv - \frac{1}{2} \delta_{a b} g^{\mu \nu} g^{\rho \sigma} \left ( \bigl ( \partial_\mu A^a_\rho \bigr ) \bigl ( \partial_\nu A^b_\sigma - \partial_\sigma A^b_\nu \bigr ) \right ) \dif V_g \\
		& \phantom{\coloneq} + \frac{\mathrm{g}}{2} \tensor{f}{_a _b _c} g^{\mu \nu} g^{\rho \sigma} \left ( \bigl ( \partial_\mu A^a_\rho \bigr ) A^b_\nu A^c_\sigma \right ) \dif V_g \\
		& \phantom{\coloneq} - \frac{\mathrm{g}^2}{4} \tensor{f}{^a _b _c} \tensor{f}{_a _d _e} g^{\mu \nu} g^{\rho \sigma} \left ( A^b_\mu A^c_\rho A^d_\nu A^e_\sigma \right ) \dif V_g
	\end{split}
	\end{equation}
	For the first example we consider a \(\operatorname{SU}(2)\) gauge group and a general spacetime manifold.\footnote{We refer to \defnssaref{def:spacetime}{defn:spacetime-matter_bundle}{defn:sheaf-of-particle-fields} for the mathematically precise definitions and to \cite[Section 2]{Prinz_5} for a geometrically more rigorous setup, using differential-graded supergeometry. \label{ftn:yang-mills-setup}} To this end, we recall that the Lie algebra is given as the three-dimensional real vector space \(\mathfrak{su}(2) \coloneq \left \langle \imaginary T^1, \imaginary T^2, \imaginary T^3 \right \rangle_\mathbb{R}\). The generators are given via \(T^a \coloneq \textfrac{\sigma^a}{2}\), where \(\sigma^a\) denotes the Pauli matrices for \(a = \set{1, 2, 3}\). Additionally, this vector space is turned into a Lie algebra by using the Levi-Civita symbol \(\tensor{\varepsilon}{^a _b _c}\) as structure constants, i.e.\ \(\tensor{f}{^a _b _c} \coloneq \tensor{\varepsilon}{^a _b _c}\) with
	\begin{subequations}
	\begin{align}
		\tensor{\varepsilon}{^1 _2 _3} & = - \tensor{\varepsilon}{^1 _3 _2} = \tensor{\varepsilon}{^2 _3 _1} = - \tensor{\varepsilon}{^2 _1 _3} = \tensor{\varepsilon}{^3 _1 _2} = - \tensor{\varepsilon}{^3 _2 _1} \coloneq 1 \, ,
		\intertext{or equivalently}
		\commutatorbig{T^2}{T^3} & = \imaginary T^1 \, , \quad \commutatorbig{T^3}{T^1} = \imaginary T^2 \quad \text{and} \quad \commutatorbig{T^1}{T^2} = \imaginary T^3 \, .
	\end{align}
	\end{subequations}
	This implies the following field content:
	\begin{equation}
		\mathcal{F}_\text{\(\operatorname{SU}(2)\)-YM} \coloneq \setbig{A_\mu^1, A_\mu^2, A_\mu^3, c^1, c^2, c^3, \overline{c}_1, \overline{c}_2, \overline{c}_3}
	\end{equation}
	With this, the expanded Yang--Mills theory Lagrange density of \eqnref{eqn:ym_lagrange_density_expanded} reads:
	\begin{equation} \label{eqn:ym-lagrange-density-expanded-su2}
	\begin{split}
		\mathcal{L}_\text{\(\operatorname{SU}(2)\)-YM} & \coloneq - \frac{1}{4 \mathrm{g}^2} g^{\mu \nu} g^{\rho \sigma} \Bigl ( F^1_{\mu \rho} F^1_{\nu \sigma} + F^2_{\mu \rho} F^2_{\nu \sigma} + F^3_{\mu \rho} F^3_{\nu \sigma} \Bigr ) \dif V_g \\
		& \phantom{:} \equiv - \frac{1}{2} g^{\mu \nu} g^{\rho \sigma} \Bigl ( \bigl ( \partial_\mu A^1_\rho \bigr ) \bigl ( \partial_\nu A^1_\sigma - \partial_\sigma A^1_\nu \bigr ) + \bigl ( \partial_\mu A^2_\rho \bigr ) \bigl ( \partial_\nu A^2_\sigma - \partial_\sigma A^2_\nu \bigr ) \\
		& \phantom{\coloneq - \frac{1}{2} g^{\mu \nu} g^{\rho \sigma} \Bigl (} + \bigl ( \partial_\mu A^3_\rho \bigr ) \bigl ( \partial_\nu A^3_\sigma - \partial_\sigma A^3_\nu \bigr ) \Bigr ) \dif V_g \\
		& \phantom{\coloneq} + \frac{\mathrm{g}}{2} g^{\mu \nu} g^{\rho \sigma} \Bigl ( \bigl ( \partial_\mu A^1_\rho \bigr ) A^2_\nu A^3_\sigma + \bigl ( \partial_\mu A^2_\rho \bigr ) A^3_\nu A^1_\sigma + \bigl ( \partial_\mu A^3_\rho \bigr ) A^1_\nu A^2_\sigma \\
		& \phantom{\coloneq + \frac{\mathrm{g}}{2} g^{\mu \nu} g^{\rho \sigma} \Bigl (} - \bigl ( \partial_\mu A^1_\rho \bigr ) A^3_\nu A^2_\sigma - \bigl ( \partial_\mu A^3_\rho \bigr ) A^2_\nu A^1_\sigma - \bigl ( \partial_\mu A^2_\rho \bigr ) A^1_\nu A^3_\sigma \Bigr ) \dif V_g \\
		& \phantom{\coloneq} - \mathrm{g}^2 g^{\mu \nu} g^{\rho \sigma} \left ( A^1_\mu A^2_\rho A^1_\nu A^2_\sigma + A^2_\mu A^3_\rho A^2_\nu A^3_\sigma + A^3_\mu A^1_\rho A^3_\nu A^1_\sigma \right ) \dif V_g
	\end{split}
	\end{equation}
	Then, the ghost homotopy Lagrange density from \eqnref{eqn:ym-gf-ghost-homotopy} specializes to the following:
	\begin{equation}
	\begin{split}
		\mathcal{L}_\text{\(\operatorname{SU}(2)\)-YM-GF-Ghost} \left ( \vartheta \right ) & \coloneq - \frac{1}{2 \mathrm{g}^2 \xi} g^{\mu \nu} \Bigl ( \bigl ( \nabla^{TM}_\mu A^1 \bigr ) \bigl ( \nabla^{TM}_\nu A^1 \bigr ) + \bigl ( \nabla^{TM}_\mu A^2 \bigr ) \bigl ( \nabla^{TM}_\nu A^2 \bigr ) \\
		& \phantom{\coloneq - \frac{1}{2 \mathrm{g}^2 \xi} g^{\mu \nu} \Bigl (} + \bigl ( \nabla^{TM}_\mu A^3 \bigr ) \bigl ( \nabla^{TM}_\nu A^3 \bigr ) \Bigr ) \dif V_g \\
		& \phantom{\coloneq} + \frac{1}{\xi} g^{\mu \nu} \left ( \bigl ( \partial_\mu \overline{c}_1 \bigr ) \bigl ( \partial_\nu c^1 \bigr ) + \bigl ( \partial_\mu \overline{c}_2 \bigr ) \bigl ( \partial_\nu c^2 \bigr ) + \bigl ( \partial_\mu \overline{c}_3 \bigr ) \bigl ( \partial_\nu c^3 \bigr ) \right ) \dif V_g \phantom{\bigg \vert} \\
		& \phantom{\coloneq} + \frac{\mathrm{g}}{2} g^{\mu \nu} \left ( 1 - \vartheta \right ) \left ( \bigl ( \partial_\mu \overline{c}_1 \bigr ) \bigl ( c^2 A^3_\nu \bigr ) + \bigl ( \partial_\mu \overline{c}_2 \bigr ) \bigl ( c^3 A^1_\nu \bigr ) + \bigl ( \partial_\mu \overline{c}_3 \bigr ) \bigl ( c^1 A^2_\nu \bigr ) \right ) \dif V_g \\
		& \phantom{\coloneq} - \frac{\mathrm{g}}{2} g^{\mu \nu} \left ( 1 - \vartheta \right ) \left (  \bigl ( \partial_\mu \overline{c}_1 \bigr ) \bigl ( c^3 A^2_\nu \bigr ) - \bigl ( \partial_\mu \overline{c}_3 \bigr ) \bigl ( c^2 A^1_\nu \bigr ) - \bigl ( \partial_\mu \overline{c}_2 \bigr ) \bigl ( c^1 A^3_\nu \bigr ) \right ) \dif V_g \\
		& \phantom{\coloneq} + \frac{\mathrm{g}}{2} g^{\mu \nu} \vartheta \left ( \bigl ( \overline{c}_1 A^2_\nu \bigr ) \bigl ( \partial_\mu c^3 \bigr ) + \bigl ( \overline{c}_2 A^3_\nu \bigr ) \bigl ( \partial_\mu c^1 \bigr ) + \bigl ( \overline{c}_3 A^1_\nu \bigr ) \bigl ( \partial_\mu c^2 \bigr ) \right ) \dif V_g \phantom{\bigg \vert} \\
		& \phantom{\coloneq} - \frac{\mathrm{g}}{2} g^{\mu \nu} \vartheta \left ( \bigl ( \overline{c}_1 A^3_\nu \bigr ) \bigl ( \partial_\mu c^2 \bigr ) + \bigl ( \overline{c}_3 A^2_\nu \bigr ) \bigl ( \partial_\mu c^1 \bigr ) + \bigl ( \overline{c}_2 A^1_\nu \bigr ) \bigl ( \partial_\mu c^3 \bigr ) \right ) \dif V_g \phantom{\bigg \vert} \\
		& \phantom{\coloneq} + \mathrm{g}^2 \xi \vartheta \left ( 1 - \vartheta \right ) \left ( \overline{c}_1 \overline{c}_2 c^1 c^2 + \overline{c}_2 \overline{c}_3 c^2 c^3 + \overline{c}_3 \overline{c}_1 c^3 c^1 \right ) \dif V_g
	\end{split}
	\end{equation}
	We emphasize that the present example shows a part of the electroweak sector of the Standard Model. Explicitly, the whole electroweak sector additionally contains the \(\operatorname{U}(1)\) gauge group for electromagnetism and the Higgs sector with its spontaneous symmetry breaking, cf.\ e.g.\ \cite{Prinz_4,Romao_Silva}.
\end{exmp}

\enter

\begin{exmp}[Schwarzschild spacetime] \label{exmp:qym-schwarzschild}
	Given the situation of \exref{exmp:qym-su-two}, for the second example we consider a Schwarzschild background and a general compact semisimple gauge group.\footnote{We refer to \ftnref{ftn:yang-mills-setup} and the references therein for the specific setup.} To this end, we use spherical coordinates \(x^\mu \equiv (t, r, \phi, \theta)\) and the \((+,-,-,-)\) sign convention for the metric. Additionally, \(G\) is Newton's constant, \(M\) the mass sitting in the coordinate origin and \(r_S \coloneq 2 G M\) the Schwarzschild radius.\footnote{We emphasize that we use units with \(c = \hbar = 1\).} Then, the Schwarzschild metric \(g^S_{\mu \nu}\) is given by
	\begin{subequations}
	\begin{align}
		g^S_{\mu \nu} \dif x^\mu \dif x^\nu & \coloneq \mathfrak{S} \dif t^2 - \mathfrak{S}^{-1} \dif r^2 - r^2 \dif \Omega^2
		\intertext{with the Schwarzschild factor}
		\mathfrak{S} & \coloneq \left ( 1 - \frac{r_S}{r} \right )
		\intertext{and the sphere volume form}
		\dif \Omega^2 & \coloneq \dif \theta^2 + \sin^2 \left ( \theta \right ) \dif \phi^2 \, .
	\end{align}
	\end{subequations}
	In addition, the corresponding non-zero components of the Christoffel symbol for the respective Levi-Civita connection are given via:
	\begin{subequations}
	\begin{align}
		& \tensor{\Gamma}{^t _t _r} \equiv \tensor{\Gamma}{^t _r _t} \coloneq - \frac{r_S}{2 r^2} \mathfrak{S}^{-1} \\
		& \tensor{\Gamma}{^r _t _t} \coloneq - \frac{r_S}{2 r^2} \mathfrak{S} && \tensor{\Gamma}{^r _r _r} \coloneq \frac{r_S}{2 r^2} \mathfrak{S}^{-1} \qquad \quad \qquad \tensor{\Gamma}{^r _\theta _\theta} \coloneq r \mathfrak{S} \qquad \tensor{\Gamma}{^r _\phi _\phi} \coloneq r \mathfrak{S} \sin^2 \left ( \theta \right ) \\
		& \tensor{\Gamma}{^\theta _r _\theta} \equiv \tensor{\Gamma}{^\theta _\theta _r} \coloneq - \frac{1}{r} && \tensor{\Gamma}{^\theta _\phi _\phi} \coloneq \frac{1}{2} \sin \left ( 2 \theta \right ) \\
		& \tensor{\Gamma}{^\phi _r _\phi} \equiv \tensor{\Gamma}{^\phi _\phi _r} \coloneq - \frac{1}{r} && \tensor{\Gamma}{^\phi _\theta _\phi} \equiv \tensor{\Gamma}{^\phi _\phi _\theta} \coloneq - \cot \left ( \theta \right )
	\end{align}
	\end{subequations}
	Finally, the Schwarzschild volume density reads
	\begin{subequations}
	\begin{align}
		\sqrt{- \dt{g_S}} & \coloneq r^2 \sin \left ( \theta \right ) \, ,
		\intertext{such that the corresponding Riemannian volume form is given by}
		\dif V_{g_S} & \coloneq r^2 \sin \left ( \theta \right ) \dif t \wedge \dif r \wedge \dif \theta \wedge \dif \phi \, .
	\end{align}
	\end{subequations}
	Inserting this into \eqnref{eqn:ym_lagrange_density_expanded} and writing \(\dif V_S \coloneq \dif t \wedge \dif r \wedge \dif \theta \wedge \dif \phi\), we obtain:
	\begin{equation} \label{eqn:ym_lagrange_density_schwarzschild}
	\begin{split}
		\mathcal{L}_\text{YM-S} & \coloneq - \frac{1}{4 \mathrm{g}^2} \delta_{a b} g_S^{\mu \nu} g_S^{\rho \sigma} F^a_{\mu \rho} F^b_{\nu \sigma} \dif V_{g_S} \\
		& \phantom{:} \equiv \frac{\sin \left ( \theta \right )}{4 \mathrm{g}^2} \delta_{a b} \Biggl ( r^2 F^a_{t r} F^b_{t r} + \mathfrak{S}^{-1} F^a_{t \theta} F^b_{t \theta} + \frac{1}{\sin^2 \left ( \theta \right ) \mathfrak{S}} F^a_{t \phi} F^b_{t \phi} \\
		& \phantom{\coloneq \frac{\sin \left ( \theta \right )}{4 \mathrm{g}^2} \delta_{a b} \Biggl (} - \mathfrak{S} F^a_{r \theta} F^b_{r \theta} - \frac{\mathfrak{S}}{\sin^2 \left ( \theta \right )} F^a_{r \phi} F^b_{r \phi} - \frac{1}{r^2 \sin^2 \left ( \theta \right )} F^a_{\theta \phi} F^b_{\theta \phi} \Biggr ) \dif V_S
	\end{split}
	\end{equation}
	Finally, the ghost homotopy Lagrange density from \eqnref{eqn:ym-gf-ghost-homotopy} specializes to the following, where \(\Lorenz^a \coloneq \mathrm{g} g^{\mu \nu} \bigl ( \nabla^{TM}_\mu A^a_\nu \bigr ) \equiv 0\) is the covariant Lorenz gauge fixing functional and \(2_\delta\) denotes the square of a Lie algebra vector \(Z^a\) with respect to the Killing form \(\delta_{ab}\), i.e.\ \(\left ( Z^a \right )^{2_\delta} \coloneq \delta_{ab} Z^a Z^b\):\footnote{We emphasize that the covariant Lorenz gauge fixing condition prevents the coupling of graviton-ghosts to gauge bosons, cf.\ \cite[Theorem 5.4]{Prinz_5}, and also constitutes an optimal choice in the sense that it only operates on gauge degrees of freedom, cf.\ \cite[Definition 3.2 and the following results]{Prinz_7}.}
	\begin{equation}
	\begin{split}
		\mathcal{L}_\text{YM-GF-Ghost-S} \left ( \vartheta \right ) & \coloneq \frac{1}{\xi} \left ( - \frac{1}{2 \mathrm{g}^2} \left ( \Lorenz_S \right )^{2_\delta} + g_S^{\mu \nu} \left ( \partial_\mu \overline{c}_a \right ) \left ( \partial_\nu c^a \right ) \right ) \dif V_{g_S} \phantom{\bigg \vert} \\
		& \phantom{\coloneq} + \frac{\mathrm{g}}{2} g_S^{\mu \nu} \tensor{f}{^a _b _c} \left ( \left ( 1 - \vartheta \right ) \bigl ( \partial_\mu \overline{c}_a \bigr ) \bigl ( c^b A^c_\nu \bigr ) + \vartheta \bigl ( \overline{c}_a A^b_\nu \bigr ) \bigl ( \partial_\mu c^c \bigr ) \right ) \dif V_{g_S} \phantom{\bigg \vert} \\
		& \phantom{\coloneq} + \frac{\mathrm{g}^2 \xi \vartheta \left ( 1 - \vartheta \right )}{4} f^{a b c} f_{ade} \overline{c}_b \overline{c}_c c^d c^e \dif V_{g_S} \phantom{\bigg \vert} \\
		& \phantom{:} \equiv - \frac{\sin \left ( \theta \right )}{2 \xi} \Biggl ( \frac{r^2}{\mathfrak{S}} \partial_t A^a_t - \left ( r^2 \mathfrak{S} \partial_r + 2 r \mathfrak{S} + r_S \right ) A^a_r \\
		& \phantom{\coloneq - \frac{\sin \left ( \theta \right )}{2 \xi} \Biggl (} - \Bigl ( \partial_\theta + \cot \left ( \theta \right ) \Bigr ) A^a_\theta - \frac{1}{\sin^2 \left ( \theta \right )} \partial_\phi A^a_\phi \Biggr )^{2_\delta} \dif V_S \\
		& \phantom{\coloneq} + \frac{\sin \left ( \theta \right )}{\xi} \Biggl ( \frac{r^2}{\mathfrak{S}} \left ( \partial_t \overline{c}_a \right ) \left ( \partial_t c^a \right ) - r^2 \mathfrak{S} \left ( \partial_r \overline{c}_a \right ) \left ( \partial_r c^a \right ) \\
		& \phantom{\coloneq + \frac{\sin \left ( \theta \right )}{\xi} \Biggl (} - \left ( \partial_\theta \overline{c}_a \right ) \left ( \partial_\theta c^a \right ) - \frac{1}{\sin^2 \left ( \theta \right )} \left ( \partial_\phi \overline{c}_a \right ) \left ( \partial_\phi c^a \right ) \Biggr ) \dif V_S \phantom{\bigg \vert} \\
		& \phantom{\coloneq} + \frac{\sin \left ( \theta \right ) \mathrm{g}}{2} \tensor{f}{^a _b _c} \Biggl ( \frac{r^2 \left ( 1 - \vartheta \right )}{\mathfrak{S}} \bigl ( \partial_t \overline{c}_a \bigr ) \bigl ( c^b A^c_t \bigr ) + \frac{r^2 \vartheta}{\mathfrak{S}} \bigl ( \overline{c}_a A^b_t \bigr ) \bigl ( \partial_t c^c \bigr ) \\
		& \phantom{\coloneq + \frac{\sin \left ( \theta \right ) \mathrm{g}}{2} \Biggl (} - r^2 \mathfrak{S} \left ( 1 - \vartheta \right ) \bigl ( \partial_r \overline{c}_a \bigr ) \bigl ( c^b A^c_r \bigr ) - r^2 \mathfrak{S} \vartheta \bigl ( \overline{c}_a A^b_r \bigr ) \bigl ( \partial_r c^c \bigr ) \\
		& \phantom{\coloneq + \frac{\sin \left ( \theta \right ) \mathrm{g}}{2} \Biggl (} - \frac{r^2 \left ( 1 - \vartheta \right )}{r^2} \bigl ( \partial_\theta \overline{c}_a \bigr ) \bigl ( c^b A^c_\theta \bigr ) - \frac{r^2 \vartheta}{r^2} \bigl ( \overline{c}_a A^b_\theta \bigr ) \bigl ( \partial_\theta c^c \bigr ) \\
		& \phantom{\coloneq + \frac{\sin \left ( \theta \right ) \mathrm{g}}{2} \Biggl (} - \frac{r^2 \left ( 1 - \vartheta \right )}{r^2 \sin^2 \left ( \theta \right )} \bigl ( \partial_\phi \overline{c}_a \bigr ) \bigl ( c^b A^c_\phi \bigr ) - \frac{r^2 \vartheta}{r^2 \sin^2 \left ( \theta \right )} \bigl ( \overline{c}_a A^b_\phi \bigr ) \bigl ( \partial_\phi c^c \bigr ) \Biggr ) \dif V_S \phantom{\bigg \vert} \\
		& \phantom{\coloneq} + \frac{r^2 \sin \left ( \theta \right ) \mathrm{g}^2 \xi \vartheta \left ( 1 - \vartheta \right )}{4} f^{a b c} f_{ade} \overline{c}_b \overline{c}_c c^d c^e \dif V_S \phantom{\bigg \vert}
	\end{split}
	\end{equation}
	We emphasize that the present example describes the situation of a stationary black hole in the origin of the coordinate system. Thus, it showcases the curved-spacetime generalization achieved in this section compared to the original flat-spacetime results of \cite{Baulieu_Thierry-Mieg}.
\end{exmp}

\section{The total construction} \label{sec:total-construction}

Combining the results from \sectionsaref{sec:case-qgr}{sec:case-qym}, we show in \thmref{thm:total_gf_ghost_lagrange_density} that the complete symmetric gauge fixing and ghost Lagrange density for the coupling of (effective) Quantum General Relativity to Quantum Yang--Mills theory can be generated via a \emph{total gauge fixing boson} using the \emph{total super-BRST operator}. Next, we observe in \colref{col:total-setup-double-homotopy} that both homotopies in the ghost construction can be added to obtain a double homotopy for the complete gauge fixing and ghost Lagrange density.

\enter

\begin{thm} \label{thm:total_gf_ghost_lagrange_density}
	We obtain the complete gauge fixing and ghost Lagrange density for (effective) Quantum General Relativity coupled to Quantum Yang--Mills theory via
	\begin{equation}
		\mathcal{L}_\textup{GR-GF-Sym-Ghost} + \mathcal{L}_\textup{YM-GF-Sym-Ghost} \equiv \mathcal{D} Y \, ,
	\end{equation}
	where \(\mathcal{D}\) is the total super-BRST operator and \(Y \coloneq F + G\) is the total gauge fixing boson.
\end{thm}

\begin{proof}
	This follows immediately from \propscref{prop:symmetric-ghost-qgr}{prop:symmetric-ghost-qym} and the following equalities:
	\begin{subequations}
	\begin{align}
		P G & \simeq_\text{TD} 0 \, , \\
		\overline{P} G & \simeq_\text{TD} 0 \, , \\
		Q F & = 0
		\intertext{and}
		\overline{Q} F & = 0 \, ,
	\end{align}
	\end{subequations}
	where \(\simeq_\text{TD}\) means equality modulo total derivatives, which hold due to \cite[Lemma 3.5]{Prinz_5}.
\end{proof}

\enter

\begin{col} \label{col:total-setup-double-homotopy}
	Given the situation of \thmref{thm:symmetric-ghost-homotopy-qgr} combined with \thmref{thm:symmetric-ghost-homotopy-qym}, we obtain the following double-homotopy in \((\varsigma, \vartheta) \in [0, 1]^2\) between the corresponding Faddeev--Popov constructions, the symmetric settings and the opposed Faddeev--Popov constructions: We apply the total BRST operator \(D \coloneq P + Q\) to the total homotopy gauge fixing fermion \(\chi \left ( \varsigma, \vartheta \right ) \coloneq \stigma \left ( \varsigma \right ) + \digamma \! \left ( \vartheta \right )\):
	\begin{equation}
		\mathcal{L}_\textup{GR-YM-GF-Ghost} \left ( \varsigma, \vartheta \right ) \coloneq D \chi \left ( \varsigma, \vartheta \right )
	\end{equation}
\end{col}

\begin{proof}
	This follows directly from \thmsaref{thm:symmetric-ghost-homotopy-qgr}{thm:symmetric-ghost-homotopy-qym} together with the argument from the proof of \thmref{thm:total_gf_ghost_lagrange_density}.
\end{proof}

\section{Conclusion} \label{sec:conclusion}

We have studied symmetric gauge fixing and ghost Lagrange densities for (effective) Quantum General Relativity coupled to Quantum Yang--Mills theory. To this end, we recalled in \sectionref{sec:general-setup} important notions of the diffeomorphism-gauge BRST double complex, introduced in \cite{Prinz_5}, together with an extensive discussion on the ghost conjugation and the shifted anti-Hermitian Lautrup--Nakanishi auxiliary fields. Thereafter, we studied the cases of (effective) Quantum General Relativity in \sectionref{sec:case-qgr} and covariant Quantum Yang--Mills theory in \sectionref{sec:case-qym}: Our results are \propsaref{prop:symmetric-ghost-qgr}{prop:symmetric-ghost-qym}, which provide the corresponding symmetric gauge fixing and ghost Lagrange densities, and \thmsaref{thm:symmetric-ghost-homotopy-qgr}{thm:symmetric-ghost-homotopy-qym}, which provide the respective homotopies between the Faddeev--Popov construction, the symmetric setting and the opposed Faddeev--Popov construction. Finally, in \sectionref{sec:total-construction} we consider the coupling of (effective) Quantum General Relativity to Quantum Yang--Mills theory: Our results are \thmref{thm:total_gf_ghost_lagrange_density}, which states that the complete symmetric gauge fixing and ghost Lagrange density can be generated from a \emph{total gauge fixing boson} via the \emph{total super-BRST operator}. In addition, we show in \colref{col:total-setup-double-homotopy} that we obtain a double homotopy if we add the corresponding homotopies of \thmsaref{thm:symmetric-ghost-homotopy-qgr}{thm:symmetric-ghost-homotopy-qym}. We want to use the symmetric ghost Lagrange densities in \cite{Prinz_7} to verify the diffeomorphism-gauge cancellation identities for (effective) Quantum General Relativity coupled to the Standard Model. This would be a major step towards the definition of a consistent renormalization operation for perturbative Quantum General Relativity in the sense of \cite{Prinz_2,Prinz_4} via the methods of \cite{Kreimer_QG1,vSuijlekom_BV,Prinz_3,Prinz_9}.

\section*{Acknowledgments}
\addcontentsline{toc}{section}{Acknowledgments}

The author thanks John Gracey for pointing out the relation to the Curci--Ferrari gauge and Jean Thierry-Mieg for suggesting an interesting follow-up project. This research is supported by the \emph{Kolleg Mathematik Physik Berlin} of the Humboldt University of Berlin and the University of Potsdam via the research group of Sylvie Paycha.

\bibliography{References}{}
\bibliographystyle{babunsrt}

\end{document}